\documentclass[runningheads, orivec]{llncs}
\usepackage[title]{appendix}
\usepackage[T1]{fontenc}
\usepackage{graphicx}
\usepackage{hyperref}
\usepackage{color}

\usepackage{amssymb}

\usepackage{tikz}
\usetikzlibrary{automata, arrows.meta, positioning, shapes.arrows}
\usetikzlibrary{backgrounds,automata}
\usetikzlibrary{calc}
\usetikzlibrary{spath3,intersections}
\usetikzlibrary{decorations.pathreplacing}

\usepackage[caption=false]{subfig}
\usepackage{mathtools}

\usepackage{paralist, tabularx}

\usepackage{wasysym}
\usepackage{trimclip}

\usepackage[shortlabels]{enumitem}

\usepackage{ntheorem}

\usepackage{placeins}

\usepackage{mymacros}
\begin{document}
\title{\papertitle}
%
%
\author{Alexander Lieb\inst{1}\and
Malte Lochau \inst{2}}
\authorrunning{A. Lieb and M. Lochau}
%
\institute{TU Darmstadt, Real-Time Systems Lab, Darmstadt, Germany \email{alexander.lieb@es.tu-darmstadt.de} \and
University of Siegen, Model-based Engineering Group, Siegen, Germany \email{malte.lochau@uni-siegen.de} 
}
\maketitle              
\begin{abstract}
%
Timed automata provide a modeling formalism for time-critical properties 
of reactive systems with discrete-state/continuous-time behaviors. 
To handle the infinite state space of timed automata, recent verification tools use zone graphs, a symbolic semantic model that guarantees sound results, at least for properties reducible to reachability problems. 
If we instead want to compare the behavior of two timed automata, checking for timed trace equivalence is undecidable. 
Fortunately, timed bisimulation equivalence is decidable, but currently available checks do not provide useful explanations of the results. 
To overcome this limitation, we use a recently proposed extension of zone graphs by so-called virtual clocks. 
The extension not only facilitates effective tool support for timed bisimilarity checking but also enables the derivation of useful explanations from the results. 
If timed bisimilarity holds, all witnesses derivable from the composed symbolic representation of both models are indeed valid for both models. 
If timed bisimilarity does not hold, we describe how to obtain counterexamples, making explicit the behavioral differences. 
These witnesses/counterexamples may serve as test cases in later stages of system refinement.
%
\end{abstract}
%


\thispagestyle{plain}
\pagestyle{plain}

\section{Introduction}
\label{chap:intro}
Timed automata are a successful formalism for reasoning about 
time-critical behaviors~\cite{Alur1990}.
Application examples include communication protocols~\cite{IEEE1394RCP} and safety-critical software~\cite{CollisionAvoidance}.
Timed automata incorporate clock variables to express quantifiable timing constraints on valid execution traces.
These variables impose infinite state spaces being computationally intractable.
Thus, most analysis tools for timed automata use zone graphs, a symbolic representation which is, however, limited to checking problems reducible to reachability~\cite{Henzinger1994}.
In contrast, a complete comparison of the timing behavior of two timed automata, based on timed trace equivalence checking is undecidable~\cite{AlurDill}, whereas timed bisimulation has been proven to be decidable~\cite{Cerans1992,Larsen1994}.
But, to make timed bisimilarity checking an effective tool that scales to realistic applications, it must be performed on zone graphs as well. 
Due to the limited precision of zone graphs, the outcomes of bisimilarity checks are neither reliable~\cite{WeiseUndLenzke1997} nor do they provide any further \emph{explanation} of the outcome~\cite{ClarkeEtAl,Pearl_2009}. 
%
The first problem is tackled in in~\cite{FullReport} by extending zone graphs with virtual clocks, while the second issue is unsolved until now, although being particularly relevant for many applications in which fine-grained insights into behavioral differences are essential. 
For example, Cort\'{e}s et al.~\cite{CortesEtAl} propose to use timed bisimulation in mutation-based testing to check whether a generated mutant of a TA is equivalent to the original model. In such a setting, a counterexample can be used to generate test cases, distinguishing the original model and the mutant \cite{MutationTesting}. Conversely, when a mutant turns out to be equivalent by accident, the developer benefits from a witness that explains why the original model and the mutant exhibit the same behavior. 

\noindent\emph{Contribution.}
In this paper, we use virtual clocks~\cite{FullReport} to facilitate the derivation of explanations of timed bisimilarity checks.
If timed bisimilarity holds, we describe how to derive a complete and finite witness representation from the composed symbolic representation of both models.
If timed bisimilarity does not hold, we describe how to obtain counterexamples, making explicit the behavioral differences. To do so, we introduce the new \emph{Product of TLTS with Virtual Clocks (PTVC)} and \emph{Product of Zone Graphs with Virtual Clocks (PZVC)}.
We provide tool support and present evaluation results gained from applying our tool to community benchmarks. 
The results show that bisimilarity checking with subsequent witness/counterexample generation scales well to larger models.

\noindent\emph{Related Work.}
Timed bisimulation originates from Moller and Tofts~\cite{Moller1990} and Yi~\cite{Yi1990}.
\u{C}er\={a}ns~\cite{Cerans1992} provides the first decidability proof using region graphs. Aceto et al.~\cite{AcetoEtAl2000} describe a formula that is fulfilled if and only if two states are timed bisimilar.
Weise and Lenzkes~\cite{WeiseUndLenzke1997} propose a more space-efficient approach using a variation of zone graphs. 
However, their descriptions lack crucial details and no tool is available. 
Tanimoto et al.~\cite{Tanimoto2004} use timed bisimulation to verify whether behavioral abstractions preserve time-critical properties, but no description of the checking procedure itself is given. 
To the best of our knowledge, the tool in~\cite{FullReport} is the only available tool for effective timed bisimulation checking. 
However, it only gives a yes/no-answer without any further explanation.
Although Wimmer et al.~\cite{wimmer2020} establish a comprehensive theory for explanations in terms of certificates for reachability analysis in timed automata, no comparable approach exists so far for timed bisimulation.

\noindent\emph{Open Science Policy.}
The open-source implementation of our approach and all examples are available at \url{https://github.com/Echtzeitsysteme/tchecker}. We also provide a Virtual Machine at \url{https://doi.org/10.5281/zenodo.20702440}.
\section{Background}
\label{sec:background}

We revisit basic definitions of timed automata~\cite{AlurDill} and timed bisimulation~\cite{Cerans1992}.
We either assume as a time domain $\TimeDomain = \mathbb{N}^{\geq 0}$ for discrete time, or $\TimeDomain = \mathbb{R}^{\geq 0}$ for dense time. 
All results apply in both settings. 
We assume a finite set $C$ of numerical variables over $\TimeDomain$, called \emph{clocks}. 
The values of all clocks in $C$ increase synchronously, but can be reset to zero independently. 
The set $\mathcal{B}(C)$ of \emph{clock constraints} $\phi$ over $C$ is defined inductively as:
%
$\phi := \text{true} \ | \ c \lesseqgtr n \ | \ c - c' \lesseqgtr n \ | \ \phi \land \phi$,
%
with $n \in \mathbb{N}^{\geq 0}$, $\lesseqgtr \ \in \{<, \leq, >, \geq\}$, $\text{ and } c, c' \in C$.
A clock constraint is \emph{atomic} if it does not contain the $\land$-operator. 
%
\begin{definition}[Timed Automaton (TA)~\cite{Henzinger1994}]
  \label{def:background:Timed-Automata:Timed-Automaton}
  A TA is a tuple $(L, l_0, \Sigma, \allowbreak C, I, E)$, where $L$ is a finite set of \emph{locations}, $l_0 \in L$ is an \emph{initial location}, $\Sigma$ is a finite set of \emph{actions}, $C$ is a finite set of \emph{clocks} such that $C \cap \Sigma = \emptyset$, $I: L \rightarrow \mathcal{B}(C)$ is a function assigning \emph{invariants} to locations, and $E \subseteq L \times \mathcal{B}(C) \times \Sigma \times 2^{C} \times L$ is a relation containing \emph{switches}.
  We write $l_1 \TATrans{g}{\sigma}{R} l_2$ for $(l_1, g, \sigma, R, l_2) \in E$.
\end{definition}
To simplify proofs, we assume TA to be \emph{diagonal-free} 
(i.e., no atomic constraint in the TA has the form $c - c' \lesseqgtr n$). 
For any non-diagonal-free TA, a language-equivalent diagonal-free TA can be constructed \cite{Berard1998}. 
We reconsider so-called difference constraints when we define zone graphs later on.

\newcommand{\backgroundTikzFontSize}{\Large}
\newcommand{\backgroundExampleAutomataScalingFactor}{0.45}
\newcommand{\backgroundExampleTAArrowDesc}{\draw[-{Latex[length=3mm]}]}

\begin{example}
  \newcommand{\backgroundExampleTABendFactor}{50}
  \newcommand{\backgroundExTAVertDistance}{1cm}
  \label{ex:background:TA}
    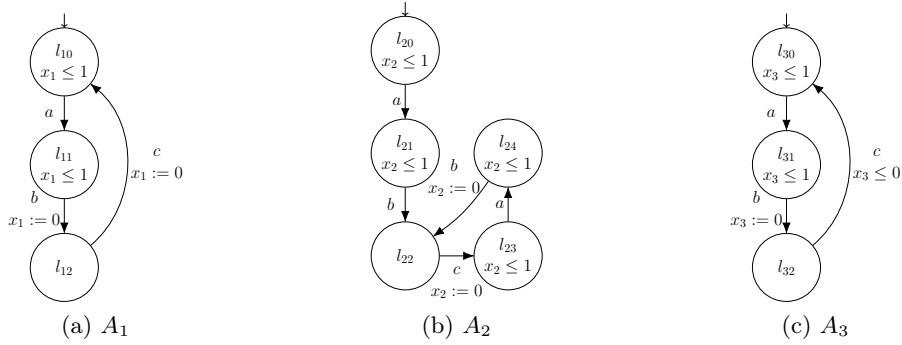
\begin{figure}[tp]
        \centering
        \subfloat[$A_1$\label{fig:exampleA1}] {
            \scalebox{\backgroundExampleAutomataScalingFactor}{
            \begin{tikzpicture}
                \tikzstyle{every node}=[font=\backgroundTikzFontSize]
                \tikzstyle{state} = [draw,circle,minimum size=2cm,inner sep=0pt,semithick]
                \node[state, align=center, initial above, initial text=] (0) {$l_{10}$\\$x_1\leq1$};
                \node[state, align=center, below = \backgroundExTAVertDistance of 0] (1a) {$l_{11}$\\$x_1\leq1$};
                \node[state, align=center, below = \backgroundExTAVertDistance of 1a] (2) {$l_{12}$};
                \backgroundExampleTAArrowDesc (0) --node[left, align=center, xshift=-0.2cm]{$a$} (1a);
                \backgroundExampleTAArrowDesc (1a) --node[left, align=center, yshift=0.2cm]{$b$\\$x_1:=0$} (2);
                \backgroundExampleTAArrowDesc (2) edge [bend right = \backgroundExampleTABendFactor] node[right, align=center] {$c$\\$x_1:=0$} (0);
            \end{tikzpicture}
            }
        }
        \hfill
        \subfloat[$A_2$\label{fig:exampleA2}]{
            \scalebox{\backgroundExampleAutomataScalingFactor}{
            \begin{tikzpicture}
                \tikzstyle{every node}=[font=\backgroundTikzFontSize]
                \tikzstyle{state} = [draw,circle,minimum size=2cm,inner sep=0pt,semithick]
                \node[state, align=center, initial above, initial text=] (0) {$l_{20}$\\$x_2\leq1$};
                \node[state, align=center, below = \backgroundExTAVertDistance of 0] (1a) {$l_{21}$\\$x_2\leq1$};
                \node[state, align=center, below = \backgroundExTAVertDistance of 1a] (2) {$l_{22}$};
                \node[state, align=center, right = \backgroundExTAVertDistance of 2] (3) {$l_{23}$\\$x_2\leq1$};
                \node[state, align=center, right = \backgroundExTAVertDistance of 1a] (4) {$l_{24}$\\$x_2\leq1$};
                \backgroundExampleTAArrowDesc (0) --node[left, align=center]{$a$} (1a);
                \backgroundExampleTAArrowDesc (1a) --node[left, align=center, xshift=-0.2cm]{$b$} (2);
                \backgroundExampleTAArrowDesc (2) --node[below, align=center, yshift=-0.2cm] {$c$\\$x_2:=0$} (3);
                \backgroundExampleTAArrowDesc (3) --node[left, align=center]{$a$} (4);
                \backgroundExampleTAArrowDesc (4) edge [bend left = 10] node[left, align=center, xshift=0.6cm, yshift=1cm]{$b$\\$x_2:=0$} (2);
            \end{tikzpicture}
            }
        }
        \hfill
        \subfloat[$A_3$\label{fig:exampleA3}]{
            \scalebox{\backgroundExampleAutomataScalingFactor}{
            \begin{tikzpicture}
                \tikzstyle{every node}=[font=\backgroundTikzFontSize]
                \tikzstyle{state} = [draw,circle,minimum size=2cm,inner sep=0pt,semithick]
                \node[state, align=center, initial above, initial text=] (0) {$l_{30}$\\$x_3\leq1$};
                \node[state, align=center, below = \backgroundExTAVertDistance of 0] (1a) {$l_{31}$\\$x_3\leq1$};
                \node[state, align=center, below = \backgroundExTAVertDistance of 1a] (2) {$l_{32}$};
                \backgroundExampleTAArrowDesc (0) --node[left, align=center, xshift=-0.2cm]{$a$} (1a);
                \backgroundExampleTAArrowDesc (1a) --node[left, align=center, yshift=0.2cm]{$b$\\$x_3:=0$} (2);
                \backgroundExampleTAArrowDesc (2) edge [bend right = \backgroundExampleTABendFactor] node[right, align=center] {$c$\\$x_3\leq0$} (0);
            \end{tikzpicture}
            }
        }
        \caption{Examples of Timed Automata}
        \label{fig:background:TA:example-TA}
    \end{figure}
    
    Figure~\ref{fig:background:TA:example-TA} shows three TA, each with one clock ($x_1$, $x_2$, or $x_3$), labeled over the same alphabet $\Sigma = \{a, b, c\}$. 
    The initial location of $A_1$ has invariant $x_1 \leq 1$ and an outgoing switch labeled \texttt{a} to $l_{11}$, which also has invariant $x_1 \leq 1$. $l_{11}$ has a single outgoing switch labeled \texttt{b} to $l_{12}$, which resets $x_1$. Finally, location $l_{12}$ has an outgoing switch to the initial location, labeled \texttt{c}, which also resets $x_1$.
    $A_2$ is similar to $A_1$ in its first two locations, except that the switch labeled \texttt{b} does not reset $x_2$. 
    From the third location, the switch labeled \texttt{c} leads to a new location $l_{23}$, which resembles the initial location with invariant $x_2 \leq 1$ and an outgoing switch labeled \texttt{a}. 
    Finally, $l_{24}$ has an outgoing switch labeled \texttt{b} to $l_{22}$, which resets $x_2$.
    $A_3$ differs from $A_1$ primarily in the naming of locations and clocks, as well as in the edge labeled \texttt{c}, which involves a guard instead of a reset.
\end{example}

Clock constraints are evaluated by \emph{clock valuations} $\ClockValuation{}: C \rightarrow \TimeDomain$, which assign a value $\ClockValuation{}(c) \in \TimeDomain$ to each clock $c \in C$. 
By $\ClockValuation{} + d$, we denote the clock valuation that assigns to each clock $c$ the value $\ClockValuation{}(c) + d$. 
For a set $R \subseteq C$, the clock valuation $\ClockValuation{}' = [R \rightarrow 0]\ClockValuation{}$ assigns to each clock $c \in R$ the value $\ClockValuation{}'(c) = 0$ and for each clock $c \in C \setminus R$ the value $\ClockValuation{}'(c) = \ClockValuation{}(c)$. 
We denote the clock valuation $u$ with $\forall c \in C : u(c) = 0$ with $[C \rightarrow 0]$. 
To check whether a clock valuation $\ClockValuation{}$ satisfies a $\phi \in \mathcal{B}(C)$, denoted $\ClockValuation{} \models \phi$, we replace all occurrences of any clock $c \in C$ in $\phi$ with its value $\ClockValuation{}(c)$ and check whether the result is true. 
We next define the operational semantics of TA as a \emph{Timed Labeled Transition Systems (TLTS)}.

\begin{definition}[TLTS]
    \label{def:background:TLTS}
    Let $A$ be a TA. The \emph{TLTS} $\TLTS{A}$ of $A$ is a tuple $(\TLTSAllStates{}, (l_0, [C \rightarrow 0]), \Sigma \cup \TimeDomain, \TLTSTrans{})$, where $\TLTSAllStates{}= L \times (C \rightarrow \TimeDomain)$ is a set of states, $(l_0, [C \rightarrow 0]) \in \TLTSAllStates{}$ is the initial state, $\Sigma \cup \TimeDomain$ is a set of transition labels, where $\Sigma \cap \TimeDomain = \emptyset$, and $\TLTSTrans{} \ \subseteq \TLTSAllStates{} \times \Sigma \cup \TimeDomain \times \TLTSAllStates{}$ (transitions) is the least relation s.t.
    \begin{compactitem}
    \item $\TLTSState{}{} \TLTSTrans{d} (l, \ClockValuation{}+d)$ if $(\ClockValuation{}+d) \models I(l)$ for $d \in \TimeDomain$, and
    \item $\TLTSState{1}{} \TLTSTrans{\sigma} \TLTSState{2}{}$ if $l_1 \TATrans{g}{\sigma}{R} l_{2}$ with $\ClockValuation{1} \models g$, $\ClockValuation{2} = [R \rightarrow 0]\ClockValuation{1}$, $\ClockValuation{2} \models I(l_2)$, and $\sigma \in \Sigma$. 
    \end{compactitem}
\end{definition}    

To compare two TA, we define \emph{timed bisimulation} on their TLTS~\cite{Lynch1992,Cerans1992}.

\begin{definition}[Timed Bisimulation]
    \label{def:background:Strong-Timed-Bisimulation}
    Let $A$ and $B$ be TA over $\Sigma$. We denote $\TLTS{A} = (\TLTSAllStates{A}, \TLTSState{A, 0}, \Sigma \cup \TimeDomain, \TLTSTrans{}_A)$ and $\TLTS{B} = (\TLTSAllStates{B}, \TLTSState{B, 0}, \Sigma \cup \TimeDomain, \TLTSTrans{}_B)$, respectively. 
    A Relation $R \subseteq \TLTSAllStates{A} \times \TLTSAllStates{B}$ is a \textit{timed bisimulation}, iff for all $(\TLTSState{A, 1}, \TLTSState{B, 1}) \in R$ it holds that 
    \begin{compactenum}
    \item if $\TLTSState{A, 1} \TLTSTrans{\mu}_A \TLTSState{A, 2}$ for $\mu \in \Sigma \cup \TimeDomain$, then there exists a $\TLTSState{B, 1} \TLTSTrans{\mu}_B \TLTSState{B, 2}$ with $(\TLTSState{A, 2}, \TLTSState{B, 2}) \in R$, and
    \item if $\TLTSState{B, 1} \TLTSTrans{\mu}_B \TLTSState{B, 2}$ for $\mu \in \Sigma \cup \TimeDomain$, then there exists a $\TLTSState{A, 1} \TLTSTrans{\mu}_A \TLTSState{A, 2}$ with $(\TLTSState{A, 2}, \TLTSState{B, 2}) \in R$.
    \end{compactenum}
    \noindent $A$ is timed bisimilar to $B$, denoted $A \sim B$, iff there exists a timed bisimulation $R$ with $(\TLTSState{A, 0}, \TLTSState{B, 0}) \in R$.
\end{definition}
\begin{example}
\label{ex:background:TLTS}
\newcommand{\backgroundExampleTLTSArrowDesc}{\draw[-{Latex[length=2mm]}]}
\newcommand{\backgroundExTLTSVertDistance}{2}
\newcommand{\backgroundExTLTSHorDistance}{2.5}
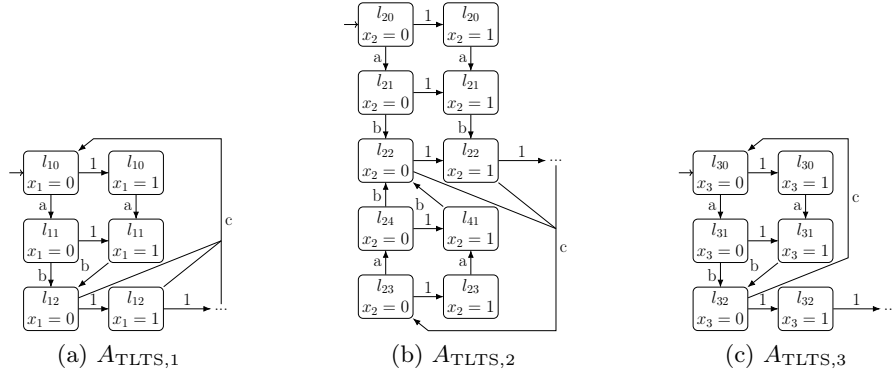
\begin{figure}[pt]
\centering
    \subfloat[$A_{\text{TLTS}, 1}$\label{subfigure:background:TLTS:example:A-1}]{
        \scalebox{\backgroundExampleAutomataScalingFactor}{
        \begin{tikzpicture}
            \tikzstyle{every node}=[font=\backgroundTikzFontSize]
            \tikzstyle{state} = [draw,rectangle,rounded corners, semithick, align=center]
            \node (l00) [state, initial, initial text = {}, align=center] at (0, 0) {$l_{10}$\\$x_1 = 0$};
            \node (l01) [state] at (\backgroundExTLTSHorDistance, 0) {$l_{10}$\\$x_1 = 1$};
            \node (l10b) [state] at (0, -\backgroundExTLTSVertDistance) {$l_{11}$\\$x_1 = 0$};
            \node (l11b) [state] at (\backgroundExTLTSHorDistance, -\backgroundExTLTSVertDistance) {$l_{11}$\\$x_1 = 1$};
            \node (l20) [state] at (0, -2*\backgroundExTLTSVertDistance) {$l_{12}$\\$x_1 = 0$};
            \node (l21) [state] at (\backgroundExTLTSHorDistance, -2*\backgroundExTLTSVertDistance) {$l_{12}$\\$x_1 = 1$};
            \node (l22) [rectangle] at (2*\backgroundExTLTSHorDistance, -2*\backgroundExTLTSVertDistance) {...};
            \backgroundExampleTLTSArrowDesc (l00) -- (l01) node[anchor=south, midway]{1};
            \backgroundExampleTLTSArrowDesc (l10b) -- (l11b) node[anchor=south, midway]{1};
            \backgroundExampleTLTSArrowDesc (l20) -- (l21) node[anchor=south, midway]{1};
            \backgroundExampleTLTSArrowDesc (l21) -- (l22) node[anchor=south, midway]{1};
            \backgroundExampleTLTSArrowDesc (l00) -- (l10b) node[anchor=east, midway]{a};
            \backgroundExampleTLTSArrowDesc (l10b) -- (l20) node[anchor=east, midway]{b};
            \backgroundExampleTLTSArrowDesc (l01) -- (l11b) node[anchor=east, midway]{a};
            \backgroundExampleTLTSArrowDesc (l11b) -- (l20) node[anchor=south east, midway]{b};
            \draw (l20) -- (2*\backgroundExTLTSHorDistance, -\backgroundExTLTSVertDistance);
            \draw (l21) -- (2*\backgroundExTLTSHorDistance, -\backgroundExTLTSVertDistance);
            \backgroundExampleTLTSArrowDesc (l22) -- node[anchor=west, midway]{c} (2*\backgroundExTLTSHorDistance, 0.5*\backgroundExTLTSVertDistance) -- (0.5*\backgroundExTLTSHorDistance, 0.5*\backgroundExTLTSVertDistance) -- (l00);
        \end{tikzpicture}
        }
    }
    \hfill
    \subfloat[$A_{\text{TLTS}, 2}$\label{subfigure:background:TLTS:example:A-2}]{
        \scalebox{\backgroundExampleAutomataScalingFactor}{
        \begin{tikzpicture}
            \tikzstyle{every node}=[font=\backgroundTikzFontSize]
            \tikzstyle{state} = [draw,rectangle,rounded corners, semithick, align=center]
            \node (l00) [state, initial, initial text = {}] at (0, 0) {$l_{20}$\\$x_2 = 0$};
            \node (l01) [state] at (\backgroundExTLTSHorDistance, 0) {$l_{20}$\\$x_2 = 1$};
            \node (l10b) [state] at (0, -\backgroundExTLTSVertDistance) {$l_{21}$\\$x_2 = 0$};
            \node (l11b) [state] at (\backgroundExTLTSHorDistance, -\backgroundExTLTSVertDistance) {$l_{21}$\\$x_2 = 1$};
            \node (l20) [state] at (0, -2*\backgroundExTLTSVertDistance) {$l_{22}$\\$x_2 = 0$};
            \node (l21) [state] at (\backgroundExTLTSHorDistance, -2*\backgroundExTLTSVertDistance) {$l_{22}$\\$x_2 = 1$};
            \node (l22) [rectangle] at (2*\backgroundExTLTSHorDistance, -2*\backgroundExTLTSVertDistance) {...};
            \node (l40) [state] at (0, -3*\backgroundExTLTSVertDistance) {$l_{24}$\\$x_2 = 0$};
            \node (l41) [state] at (\backgroundExTLTSHorDistance, -3*\backgroundExTLTSVertDistance) {$l_{41}$\\$x_2 = 1$};  
            \node (l30) [state] at (0, -4*\backgroundExTLTSVertDistance) {$l_{23}$\\$x_2 = 0$};
            \node (l31) [state] at (\backgroundExTLTSHorDistance, -4*\backgroundExTLTSVertDistance) {$l_{23}$\\$x_2 = 1$};
            \backgroundExampleTLTSArrowDesc (l00) -- (l01) node[anchor=south, midway]{1};
            \backgroundExampleTLTSArrowDesc (l10b) -- (l11b) node[anchor=south, midway]{1};
            \backgroundExampleTLTSArrowDesc (l20) -- (l21) node[anchor=south, midway]{1};
            \backgroundExampleTLTSArrowDesc (l21) -- (l22) node[anchor=south, midway]{1};
            \backgroundExampleTLTSArrowDesc (l30) -- (l31) node[anchor=south, midway]{1};
            \backgroundExampleTLTSArrowDesc (l40) -- (l41) node[anchor=south, midway]{1};
            \backgroundExampleTLTSArrowDesc (l00) -- (l10b) node[anchor=east, midway]{a};
            \backgroundExampleTLTSArrowDesc (l10b) -- (l20) node[anchor=east, midway]{b};
            \backgroundExampleTLTSArrowDesc (l01) -- (l11b) node[anchor=east, midway]{a};
            \backgroundExampleTLTSArrowDesc (l11b) -- (l21) node[anchor=east, midway]{b};
            \backgroundExampleTLTSArrowDesc (l30) -- (l40) node[anchor=east, midway]{a};
            \backgroundExampleTLTSArrowDesc (l40) -- (l20) node[anchor=east, midway]{b};
            \backgroundExampleTLTSArrowDesc (l31) -- (l41) node[anchor=east, midway]{a};
            \backgroundExampleTLTSArrowDesc (l41) -- (l20) node[anchor=north east, midway]{b};
            \backgroundExampleTLTSArrowDesc (l22) -- node[anchor=west, midway]{c} (2*\backgroundExTLTSHorDistance, -4.5*\backgroundExTLTSVertDistance) -- (0.5*\backgroundExTLTSHorDistance, -4.5*\backgroundExTLTSVertDistance) -- (l30);
            \draw (l20) -- (2*\backgroundExTLTSHorDistance, -3*\backgroundExTLTSVertDistance);
            \draw (l21) -- (2*\backgroundExTLTSHorDistance, -3*\backgroundExTLTSVertDistance);
        \end{tikzpicture}
        }
    }
    \hfill
    \subfloat[$A_{\text{TLTS}, 3}$\label{subfigure:background:TLTS:example:A-3}]{
        \scalebox{\backgroundExampleAutomataScalingFactor}{
        \begin{tikzpicture}
            \tikzstyle{every node}=[font=\backgroundTikzFontSize]
            \tikzstyle{state} = [draw,rectangle,rounded corners, semithick, align=center]
            \node (l00) [state, initial, initial text = {}] at (0, 0) {$l_{30}$\\$x_3 = 0$};
            \node (l01) [state] at (\backgroundExTLTSHorDistance, 0) {$l_{30}$\\$x_3 = 1$};
            \node (l10b) [state] at (0, -\backgroundExTLTSVertDistance) {$l_{31}$\\$x_3 = 0$};
            \node (l11b) [state] at (\backgroundExTLTSHorDistance, -\backgroundExTLTSVertDistance) {$l_{31}$\\$x_3 = 1$};
            \node (l20) [state] at (0, -2*\backgroundExTLTSVertDistance) {$l_{32}$\\$x_3 = 0$};
            \node (l21) [state] at (\backgroundExTLTSHorDistance, -2*\backgroundExTLTSVertDistance) {$l_{32}$\\$x_3 = 1$};
            \node (l22) [rectangle] at (2*\backgroundExTLTSHorDistance, -2*\backgroundExTLTSVertDistance) {...};
            \backgroundExampleTLTSArrowDesc (l00) -- (l01) node[anchor=south, midway]{1};
            \backgroundExampleTLTSArrowDesc (l10b) -- (l11b) node[anchor=south, midway]{1};
            \backgroundExampleTLTSArrowDesc (l20) -- (l21) node[anchor=south, midway]{1};
            \backgroundExampleTLTSArrowDesc (l21) -- (l22) node[anchor=south, midway]{1};
            \backgroundExampleTLTSArrowDesc (l00) -- (l10b) node[anchor=east, midway]{a};
            \backgroundExampleTLTSArrowDesc (l10b) -- (l20) node[anchor=east, midway]{b};
            \backgroundExampleTLTSArrowDesc (l01) -- (l11b) node[anchor=east, midway]{a};
            \backgroundExampleTLTSArrowDesc (l11b) -- (l20) node[anchor=south east, midway]{b};
            \path[spath/save=cross 2] (l20) -- (1.5*\backgroundExTLTSHorDistance, -1.25*\backgroundExTLTSVertDistance) -- node[anchor=west, midway]{c} (1.5*\backgroundExTLTSHorDistance, 0.5*\backgroundExTLTSVertDistance) -- (0.5*\backgroundExTLTSHorDistance, 0.5*\backgroundExTLTSVertDistance) -- (l00);
            \backgroundExampleTLTSArrowDesc[spath/use=cross 2];
      \end{tikzpicture}
      }
    }
\caption{Extracts of the TLTS of the TA of Figure \ref{fig:background:TA:example-TA}}
\label{fig:background:TLTS:example}
\end{figure}

We compare the TA from Example~\ref{ex:background:TA} by considering the excerpt of their respective TLTS, shown in Figure~\ref{fig:background:TLTS:example}.
In $A_{\text{TLTS}, 1}$, the initial state $(l_{10}, x_1 = 0)$ 
can delay up to one time unit (due to the invariant of $l_{10}$) 
or take the transition labeled \texttt{a}, followed by a delay or \texttt{b}. 
At location $l_{12}$, infinite delays are possible, or a transition labeled \texttt{c} can be used to reach the initial state again. 
Analogously, $A_{\text{TLTS}, 2}$ has an equivalent beginning, except that there is no reset at the transition labeled \texttt{b} and the state $(l_{21}, x_2 = 1)$ has an outgoing transition to $(l_{22}, x_2 = 1)$. 
The part concerning $l_{23}$ and $l_{24}$ of $A_{\text{TLTS}, 2}$ is analogous to $A_{\text{TLTS}, 1}$. $A_{\text{TLTS}, 3}$ is equivalent to $A_{\text{TLTS}, 1}$, but has only one state with an outgoing transition labeled \texttt{c}.
To show $A_1 \sim A_2$, we construct a timed bisimulation $R_{1, 2}$ with $((l_{10}, x_1 = 0), (l_{20}, x_2 = 0)) \in R_{1, 2}$. 
By Def.~\ref{def:background:Strong-Timed-Bisimulation}, this implies $((l_{10}, x_1 = 1), (l_{20}, x_2 = 1)) \in R_{1, 2}$, $((l_{11}, x_1 = 0), (l_{21}, x_2 = 0)) \in R_{1, 2}$, and $((l_{11}, x_1 = 1), (l_{21}, x_2 = 1)) \in R_{1, 2}$.
Iteratively applying this procedure confirms that $R_{1, 2}$ fulfills the conditions of Def.~\ref{def:background:Strong-Timed-Bisimulation}. 
In contrast, assuming a timed bisimulation $R_{1, 3}$ with $((l_{10}, x_1 = 0), (l_{30}, x_3 = 0)) \in R_{1, 3}$ implies $((l_{12}, x_1 = 1), (l_{32}, x_3 = 1)) \in R_{1, 3}$, which contradicts the assumption that $R_{1, 3}$ is a timed bisimulation since $(l_{12}, x_1 = 1)$ has a transition labeled \texttt{c}, which is not the case for $(l_{32}, x_3 = 1)$. 
Hence, no timed bisimulation that contains $((l_{10}, x_1 = 0), (l_{30}, x_3 = 0))$ exists and $A_1 \not\sim A_3$. The reader might recognizes that the actual TLTS are infinitely large (even uncountable in dense time model).
\end{example}

Since TLTS are infinitely large they are not tractable by analysis tools. 
This problem is solved by symbolic TA semantics. 
A \emph{symbolic state} consists of a location and a \emph{zone}. 
A zone comprises a set of clock valuations, where we denote $2^{C \rightarrow \TimeDomain} = \mathcal{D}(C)$. $\TLTSState{A} \in \ZGState{B}$ holds if and only if $l_A = l_B$ and $\ClockValuation{A} \in \Zone{B}$. 
A zone $\Zone{} \in \mathcal{D}(C)$ is typically represented by a clock constraint $\phi\in\mathcal{B}(C)$, such that $\Zone{} = \{\ClockValuation{} \ | \ \ClockValuation{} \in (C \rightarrow \TimeDomain) \wedge \ClockValuation{} \models \phi\}$. 
Assume a zone $\Zone{}\in \mathcal{D}(C)$, a clock constraint $\phi_{cc} \in \mathcal{B}(C)$, and any $R \subseteq C$. We use the operations $\Zone{} \land \phi_{cc} = \{u \ | \ \ClockValuation{} \in \Zone{} \land \ClockValuation{} \models \phi_{cc}\}$ ($\land$-operation), $\Zone{}^{\uparrow} = \{\ClockValuation{}+d \ | \ d \in \TimeDomain \land \ClockValuation{} \in \Zone{}\}$ (future-operation), and $R(\Zone{}) = \{[R \rightarrow 0] \ClockValuation{} \ | \ \ClockValuation{} \in \Zone{}\}$ (reset-operation).
The following definition is a slightly adapted definition of the zone graphs of Bengtsson and Yi~\cite{Bengtsson2004}. A transition labeled $\varepsilon \not\in \Sigma$ is used to represent delay transitions~\cite{Yi1995}.

\begin{definition}[Zone Graph]
\label{def:background:zonegraphs}
Let $A$ be a TA. 
A \emph{zone graph} $\ZG{A} = (\ZGAllStates{}, \ZGState{0}{}, \Sigma \cup \{\varepsilon\}, \ZGTrans{})$ 
of $A$ is a transition system, where $\ZGAllStates{} = L \times \mathcal{D}(C)$ is the set of symbolic states, $(l_0, \{[C \rightarrow 0]\}) \in \ZGAllStates{}$ is the initial symbolic state, and $\ZGTrans{} \ \subseteq \ZGAllStates{} \times \Sigma \cup \{\varepsilon\} \times \ZGAllStates{}$ is the least relation satisfying the rules:
\begin{compactitem}
    \item $\ZGState{}{} \ZGTrans{\varepsilon} (l, \Zone{}^{\uparrow} \land I(l))$, and
    \item $\ZGState{1}{} \ZGTrans{\sigma} (l_2, R(\Zone{1} \land g) \land I(l_2))$, if $l_1 \TATrans{g}{\sigma}{R} l_{2}$ and $R(\Zone{1} \land g) \land I(l_2) \neq \emptyset$.
\end{compactitem}
\end{definition}
While zone graphs can be used for checking reachability of TA states, they do not preserve enough information 
of the TA to check for timed bisimulation~\cite{WeiseUndLenzke1997}.
\begin{example}
\label{ex:background:ZG}
\newcommand{\backgroundExampleZoneGraphScalingFactor}{0.45}
\newcommand{\backgroundExZGHorDistance}{2}
\newcommand{\backgroundExZGVertDistance}{1}
\begin{figure}[tp]
    \centering
    \scalebox{\backgroundExampleZoneGraphScalingFactor}{
    \begin{tikzpicture}
        \tikzstyle{every node}=[font=\backgroundTikzFontSize]
        \tikzstyle{symstate} = [draw,rectangle,minimum width=3.5cm,minimum height=1.5cm,inner sep=5pt,thick]
        %
        \node[symstate, align=center, initial left, initial text=] (00) {$l_{10}$\\$x_1=0$};
        \node[symstate, align=center, right = \backgroundExZGVertDistance of 00] (01) {$l_{10}$\\$x_1\leq1$};
        %
        \node[symstate, align=center, right = \backgroundExZGVertDistance of 01] (10) {$l_{11}$\\$x_1=0$};
        \node[symstate, align=center, right = \backgroundExZGVertDistance of 10] (11) {$l_{11}$\\$x_1\leq1$};
        %
        \node[symstate, align=center, right = \backgroundExZGVertDistance of 11] (20) {$l_{12}$\\$x_1=0$};
        \node[symstate, align=center, right = \backgroundExZGVertDistance of 20] (21) {$l_{12}$\\$\text{true}$};
        \backgroundExampleTAArrowDesc (00) --node[above, align=center]{$\varepsilon$} (01);
        \backgroundExampleTAArrowDesc (01) to[loop below] node[above, align=center]{$\varepsilon$} (01);
        \backgroundExampleTAArrowDesc (10) --node[above, align=center]{$\varepsilon$} (11);
        \backgroundExampleTAArrowDesc (11) to[loop below] node[above, align=center]{$\varepsilon$} (11);
        \backgroundExampleTAArrowDesc (20) --node[above, align=center]{$\varepsilon$} (21);
        \backgroundExampleTAArrowDesc (21) to[loop below] node[above, align=center]{$\varepsilon$} (21);
        %
        %
        \node[state, minimum size=0, inner sep=0, draw = none, above right = 0.5*\backgroundExZGVertDistance and 0.3*\backgroundExZGHorDistance of 00] (from00to10first){};
        \node[state, minimum size=0, inner sep=0, draw = none, above left = 0.5*\backgroundExZGVertDistance and 0.3*\backgroundExZGHorDistance of 10] (from00to10second){};
        \path[spath/save=cross] ($(00.north east) - (0.1*\backgroundExZGHorDistance, 0)$) |- (from00to10first) -- node[below, align=center]{a} (from00to10second) -| ($(10.north west) + (0.1*\backgroundExZGHorDistance, 0)$);
        \backgroundExampleTAArrowDesc[spath/use=cross];
        \node[state, minimum size=0, inner sep=0, draw = none, below right = 0.4*\backgroundExZGVertDistance and 0.3*\backgroundExZGHorDistance of 01] (from01to11first){};
        \node[state, minimum size=0, inner sep=0, draw = none, below left = 0.4*\backgroundExZGVertDistance and 0.3*\backgroundExZGHorDistance of 11] (from01to11second){};
        \path[spath/save=cross3] ($(01.south east) - (0.1*\backgroundExZGHorDistance, 0)$) |- (from01to11first) -- node[below, align=center]{a} (from01to11second) -| ($(11.south west) + (0.1*\backgroundExZGHorDistance, 0)$);
        \backgroundExampleTAArrowDesc[spath/use=cross3];
        %
        \node[state, minimum size=0, inner sep=0, draw = none, above right = 0.5*\backgroundExZGVertDistance and 0.3*\backgroundExZGHorDistance of 10] (from10to20first){};
        \node[state, minimum size=0, inner sep=0, draw = none, above left = 0.5*\backgroundExZGVertDistance and 0.3*\backgroundExZGHorDistance of 20] (from10to20second){};
        \path[spath/save=cross4] ($(10.north east) - (0.1*\backgroundExZGHorDistance, 0)$) |- (from10to20first) -- node[below, align=center]{b} (from10to20second) -| ($(20.north west) + (0.1*\backgroundExZGHorDistance, 0)$);
        \backgroundExampleTAArrowDesc[spath/use=cross4];
        \backgroundExampleTAArrowDesc (11) --node[above, align=center]{b} (20);
        %
        \node[state, minimum size=0, inner sep=0, draw = none, above = 0.7*\backgroundExZGVertDistance of 00] (above00){};
        \path[spath/save=cross 2] (21) |- (above00) -| (00);
        \backgroundExampleTAArrowDesc[spath/use=cross 2];
        \draw (20) |-node[above, align=center]{c} (above00);
      \end{tikzpicture}
    }
    \caption{Zone Graph of $A_1$.}
    \label{fig:background:example-zg}
\end{figure}
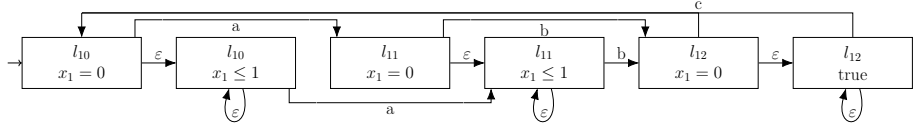 
Applying Def.~\ref{def:background:zonegraphs} to $A_1$ yields the zone graph in Fig.~\ref{fig:background:example-zg}. 
$A_3$ has the exact same zone graph up-to renaming of locations and the clock although $A_1$ and $A_3$ not bisimilar.
\end{example}
Zone graphs lack two important properties 
for bisimulation checking: (1) prone to false positives due to
subtle effects of clock resets and (2) they do not contain sufficient information to explain the outcome. 
Issue (1) is addressed in~\cite{FullReport} by adding \emph{virtual clocks} into zone graphs, whereas (2) is addressed in the remainder.
\section{Finding Witnesses}
\label{chap:main:finding-witnesses}

%
%
We describe how to represent in a finite way 
a timed bisimulation between two TA.
For this, we first describe a set-based construction of 
bisimulation witnesses in a single unified graph based on
\emph{synchronous product of TLTS}~\cite{Nicollin1994}.
We extend the product by a set of \emph{virtual clocks}, called PTVC.
Let $C_{AB}$ be the disjoint union of the set clocks $A$, $B$ of two TA. 
Each original clock $c \in C_{AB}$ is paired with exactly one virtual clock $\chi(c) \in \chi(C_{AB})$ assuming $\chi(C_{AB}) \cap C_{AB} = \emptyset$, which initially has the same interpretation as its original clock. 
We decompose action transitions into two steps. 
First, only the original clocks, which are reset at action transitions of the TA, are also reset in the PTVC. 
By leaving the virtual clocks unaffected, they preserve the missing information 
required to correctly distinguish non-bisimilar states (see Example~\ref{ex:background:ZG}).
Second, virtual clocks are reset before 
the next transition takes place to synchronize with their original clocks. 

\newcommand{\newdefTikzFontSize}{\Large}
\begin{example}
  \label{ex:main:virtual-bisimulation:init-example}
  \newcommand{\newdefExampleProductScalingFactor}{0.45}
  \newcommand{\newdefProdHorDistance}{1.2}
  \newcommand{\newdefProdVertDistance}{0.5}
  \newcommand{\newdefExampleProdArrowDesc}{\draw[-{Latex[length=2mm]}]}
  \begin{figure}[tp]
    \centering
    \scalebox{\newdefExampleProductScalingFactor}{
      \begin{tikzpicture}
        \tikzstyle{every node}=[font=\newdefTikzFontSize]
        \tikzstyle{state} = [draw,rectangle,rounded corners, inner sep=1pt, text width=4.7cm, minimum height=1.9cm, semithick, align=center]
        \node[state, initial left, initial text=] (10300) {$l_{10}, l_{30}$\\$x_{1}=x_{3}=$ \\ $\chi(x_1)=\chi(x_3)=0$};
        \node[state, right = \newdefProdHorDistance of 10300] (10301) {$l_{10}, l_{30}$\\$x_{1}=x_{3}=$ \\ $\chi(x_1)=\chi(x_3)=1$};
        \node[state, right = \newdefProdHorDistance of 10301] (11310) {$l_{11}, l_{31}$\\$x_{1}=x_{3}=1$ \\ $\chi(x_1)=\chi(x_3)=1$};
        \node[state, right = \newdefProdHorDistance of 11310] (11311) {$l_{11}, l_{31}$\\$x_{1}=x_{3}=$ \\ $\chi(x_1)=\chi(x_3)=0$};
        \node[state, below = \newdefProdVertDistance of 10301] (12321) {$l_{12}, l_{32}$\\$x_{1}=x_{3}=0$ \\ $\chi(x_1)=\chi(x_3)=1$};
        \node[state, right = \newdefProdHorDistance of 12321] (12320) {$l_{12}, l_{32}$\\$x_{1}=x_{3}=0$ \\ $\chi(x_1)=\chi(x_3)=0$};
        \node[state, right = \newdefProdHorDistance of 12320] (12322) {$l_{12}, l_{32}$\\$x_{1}=x_{3}=1$ \\ $\chi(x_1)=\chi(x_3)=1$};
        \node[rectangle, align=center, right = \newdefProdHorDistance of 12322] (12323) {\dots};
        %
        \newdefExampleProdArrowDesc (10300) --node[above, align=center]{1} (10301);
        \node[rectangle, minimum size=0, inner sep=0, draw = none, above = 1*\newdefProdVertDistance of 10300] (above10300){};
        \node[rectangle, minimum size=0, inner sep=0, draw = none, above = 1*\newdefProdVertDistance of 11311] (above11311){};
        \path[spath/save=cross1] (10300) |- (above10300) --node[above, align=center]{a} (above11311) -|  (11311);
        \backgroundExampleTAArrowDesc[spath/use=cross1];
        \newdefExampleProdArrowDesc (10301) --node[above, align=center]{a} (11310);
        %
        \newdefExampleProdArrowDesc (11311) --node[above, align=center]{1} (11310);
        \newdefExampleProdArrowDesc (11311) --node[right, align=center, xshift=0.5cm]{b} (12320);
        \newdefExampleProdArrowDesc (11310) --node[right, align=center, xshift=0.5cm]{b} (12321);
        %
        \newdefExampleProdArrowDesc (12321) --node[above, align=center]{sync} (12320);
        \newdefExampleProdArrowDesc (12322) --node[above, align=center]{1} (12323);
        \newdefExampleProdArrowDesc (12320) --node[above, align=center]{1} (12322);
        %
        \node[rectangle, minimum size=0, inner sep=0, draw = none, below = 0.5*\newdefProdVertDistance of 12320] (below12320){};
        \path[spath/save=cross2] (12320) |- (below12320) -|node[above, xshift=10*\newdefProdHorDistance, align=center]{c} (10300);
        \backgroundExampleTAArrowDesc[spath/use=cross2];
      \end{tikzpicture}
    }
  \caption{\PTVC{} of $A_1$ and $A_3$.}
  \label{fig:main:virtual-bisimulation:product-synchronizable-tlts:example}
  \end{figure}
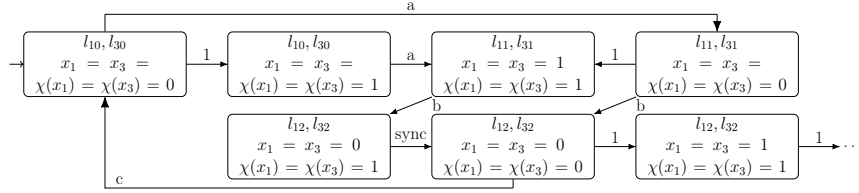
  Figure~\ref{fig:main:virtual-bisimulation:product-synchronizable-tlts:example} shows the \PTVC{} of $A_1$ and $A_3$ (only delay transitions for $d=1$). 
  In the initial state, all clocks are set to zero. 
  As both TA have enabled \texttt{a}-transitions, the initial state has an \texttt{a}-transition to the state with $(l_{11}, l_{31})$. 
  A $1$-delay is also possible in both TA, resulting in a \texttt{1}-labeled transition
  after which both TA, again, have an \texttt{a}-transition. 
  This \texttt{a}-transition leads to a state in which only a \texttt{b}-transition is enabled.
  In the next state, the original clocks are set to zero due to the clock resets of the switches while the virtual clock values remain unchanged. 
  Next, a \texttt{sync}-transition is added which targets the state in which the values of the virtual clocks are also set to zero. 
  After another delay, we reach the state with location pair $(l_{12}, l_{32})$ and all clocks are set to $1$. 
  Here, $A_1$ enables a \texttt{c}-transition, but $A_3$ does not, so the state of the product has no \texttt{c}-transition.
\end{example}


\begin{definition}[\PTVC{}]
  \label{def:main:virtual-bisimulation:product-synchronizable-tlts}
  The \PTVC{} of two TA $A$ and $B$ is a labeled transition system $(\STLTSAllStates{}, \allowbreak s_{AB, 0}, \Sigma \cup \TimeDomain \cup \{\text{sync}\}, \STLTSTrans{})$. $\STLTSAllStates{} \subseteq (L_A \times L_B) \times (C_{AB} \cup \chi(C_{AB}) \rightarrow \TimeDomain)$ is the set of states $((l_A, l_B), \ClockValuation{})$ with $\forall c \in C_{AB} : (\ClockValuation{}(c) = \ClockValuation{}(\chi(c)) \lor \ClockValuation{}(c) = 0)$. We denote $\STLTSAllStates{s} \subseteq \STLTSAllStates{}$ for the set of all states $((l_A, l_B), \ClockValuation{s}) \in \STLTSAllStates{}$ with $\forall c \in C_{AB} : \ClockValuation{}(c) = \ClockValuation{}(\chi(c))$. Moreover, $s_{AB, 0} = ((l_{A, 0}, l_{B, 0}), \allowbreak [C_{AB} \cup \chi(C_{AB}) \rightarrow 0]) \in \STLTSAllStates{s}$ is the initial state, $\Sigma \cup \TimeDomain \cup \{\text{sync}\}$ is the set of transition labels, and $\STLTSTrans{} \ \subseteq \STLTSAllStates{} \times \Sigma \cup \TimeDomain \cup \{\text{sync}\} \times \STLTSAllStates{}$ is a set of transitions, which is the least relation satisfying
  \begin{compactitem}
  \item $((l_A, l_B), \ClockValuation{}) \STLTSTrans{\text{sync}} ((l_A, l_B), \ClockValuation{\text{sync}})$, if $((l_A, l_B), \ClockValuation{}) \not\in \STLTSAllStates{s}$, $\forall c \in C_{AB} : \ClockValuation{}(c) = \ClockValuation{\text{sync}}(c)$, and $((l_A, l_B), \ClockValuation{\text{sync}}) \in \STLTSAllStates{s}$,
  \item $((l_A, l_B), \ClockValuation{}) \STLTSTrans{d} ((l_A, l_B), \ClockValuation{}+d)$ if $((l_A, l_B), \ClockValuation{}) \in \STLTSAllStates{s}$, $(\ClockValuation{}+d) \models I(l_A)$, and $(\ClockValuation{} + d) \models I(l_B)$ for $d \in \TimeDomain$, and
  \item $((l_A, l_B), \ClockValuation{}) \STLTSTrans{\sigma} ((l_{A, 1}, l_{B, 1}), \ClockValuation{1})$ if $((l_A, l_B), \ClockValuation{}) \in \STLTSAllStates{s}$, $l_A \TATrans{g_A}{\sigma}{R_A} l_{A, 1}$ with $\ClockValuation{} \models g_A$, $l_B \TATrans{g_B}{\sigma}{R_B} l_{B, 1}$ with $\ClockValuation{} \models g_B$, $\ClockValuation{1} = [R_A \cup R_B \rightarrow 0]\ClockValuation{}$, $\ClockValuation{1} \models I(l_{A, 1})$,  $\ClockValuation{1} \models I(l_{B, 1})$, and $\sigma \in \Sigma$. 
  \end{compactitem}
\end{definition}
We call a state $((l_A, l_B), \ClockValuation{}) \in \STLTSAllStates{s}$ \textit{synchronized}.
The target of a delay transition of a synchronized state is also synchronized, as delay transitions cannot perform clock resets, while the target of an action transition might not be synchronized due to resets.
Based on \PTVC, we define \emph{virtual bisimulation}.

\begin{definition}[Virtual Bisimulation]
  \label{def:main:virtual-bisimulation:virtual-bisimulation}
  Given two TA $A$, $B$, with their \PTVC{} $(\STLTSAllStates{}, s_{AB, 0}, \Sigma \cup \TimeDomain \cup \{\text{sync}\}, \STLTSTrans{})$. A relation $R_{\text{virt}} \subseteq \STLTSAllStates{s}$ is called a virtual bisimulation, if and only if for all elements $((l_A, l_B), \ClockValuation{}) \in R_{\text{virt}}$ it holds that 
  \begin{compactenum}
    \item for any $d \in \TimeDomain$ with $\ClockValuation{} + d \models I(l_A)$, there exists a transition $((l_A, l_B), \ClockValuation{}) \STLTSTrans{d} ((l_A, l_B), \ClockValuation{} + d)$ and $((l_A, l_B), \ClockValuation{} + d) \in R_{\text{virt}}$,
    \item if there exists a transition $l_A \TATrans{g_A}{\sigma}{R_A} l_{A, 1}$ with $\ClockValuation{} \models g_A$ and $[R_A \rightarrow 0] \ClockValuation{} \models I(l_{A, 1})$, then there exists a transition $((l_A, l_B), \ClockValuation{}) \STLTSTrans{\sigma} ((l_{A, 1}, l_{B, 1}), \ClockValuation{1})$ with $\forall c \in C_A : [R_A \rightarrow 0] \ClockValuation{}(c) = \ClockValuation{1}(c)$ and either
    \begin{enumerate}
      \item $((l_{A, 1}, l_{B, 1}), \ClockValuation{1})$ is synchronized and $((l_{A, 1}, l_{B, 1}), \ClockValuation{1}) \in R_{\text{virt}}$ holds, or
      \item $((l_{A, 1}, l_{B, 1}), \ClockValuation{1})$ is not synchronized and there exists a transition $((l_{A, 1}, \allowbreak l_{B, 1}), \ClockValuation{1}) \STLTSTrans{\text{sync}} ((l_{A, 1}, l_{B, 1}), \ClockValuation{1, s})$ with $((l_{A, 1}, l_{B, 1}), \ClockValuation{1, s}) \in R_{\text{virt}}$, and
    \end{enumerate}
    \item the analog for $B$.
  \end{compactenum}
\end{definition}
TA $A$ and $B$ are \emph{virtually bisimilar} if and only if there exists a virtual bisimulation containing the initial state of their \PTVC{}. 
Moreover, two TA are virtually bisimilar if and only if they are timed bisimilar according to Def.~\ref{def:background:Strong-Timed-Bisimulation}, where the mapping from a virtual bisimulation to a timed bisimulation is straightforward (see Appendix~\ref{app:equi-def-vb}). 
As a virtual bisimulation may, again, contain infinitely many state pairs, we 
use a \emph{symbolic} representation. 
Given a set of symbolic states $S$, we denote $\text{states}(S) = \{ \STLTSState{} \ | \ \exists \SZGState{} \in S :  \ClockValuation{} \in \Zone{}\}$. 
We next describe how to find a set $S$ of symbolic states such that $\text{states}(S)$ forms a virtual bisimulation. 
To do so, we construct the product of zone graphs with virtual clocks (\PZVC).

\newcommand{\prodSZGScalingFactor}{0.45}
\newcommand{\prodSZGHorDistance}{1.2}
\newcommand{\prodSZGVertDistance}{0.5}
\newcommand{\prodSZGTikzFontSize}{\Large} 

\begin{example}
  \label{ex:main:finding-wit-and-counter:init-example-PZVC}
  \begin{figure}[tp]
    \centering
    \scalebox{\prodSZGScalingFactor}{
    \begin{tikzpicture}
      \tikzstyle{every node}=[font=\prodSZGTikzFontSize]
      \tikzstyle{symstate} = [draw,rectangle,minimum width=5.4cm,minimum height=2cm,inner sep=5pt,thick]
      %
      \node[symstate, align=center, initial left, initial text=] (00) {$(l_{10}, l_{30})$, $x_1=x_3=0$\\$\chi(x_1)=\chi(x_3)=0$};
      \node[symstate, align=center, below = \prodSZGVertDistance of 00] (01) {$(l_{10}, l_{30})$, $x_1=x_3\leq1$\\$\chi(x_1)=x_1, \chi(x_3)=x_3$};
      %
      \node[symstate, align=center, right = \prodSZGHorDistance of 00] (10) {$(l_{11}, l_{31})$, $x_1=x_3=0$\\$\chi(x_1)=x_1, \chi(x_3)=x_3$};
      \node[symstate, align=center, right = \prodSZGHorDistance of 01] (11) {$(l_{11}, l_{31})$, $x_1=x_3\leq1$\\$\chi(x_1)=x_1, \chi(x_3)=x_3$};
      %
      \node[symstate, align=center, right = \prodSZGHorDistance of 10] (20) {$(l_{12}, l_{32})$, $x_1=0, x_3=0$\\$\chi(x_1)=\chi(x_3)=x_3$};
      \node[symstate, align=center, right = \prodSZGHorDistance of 20] (20e) {$(l_{12}, l_{32})$, $x_1=x_3$\\$\chi(x_1)=\chi(x_3)=x_3$};
      \node[symstate, align=center, right = \prodSZGHorDistance of 11] (21) {$(l_{12}, l_{32})$, $x_1=0, x_3 = 0$\\$\chi(x_1)=\chi(x_3)\leq 1$};
      %
      %
      \backgroundExampleTAArrowDesc (00) --node[right, align=center, xshift=0.2cm]{$\varepsilon$} (01);
      \backgroundExampleTAArrowDesc (10) --node[right, align=center, xshift=0.2cm]{$\varepsilon$} (11);
      \backgroundExampleTAArrowDesc (20) --node[above, align=center]{$\varepsilon$} (20e);
      %
      \backgroundExampleTAArrowDesc (01) to[loop below] node[above, align=center]{$\varepsilon$} (01);
      \backgroundExampleTAArrowDesc (11) to[loop below] node[above, align=center]{$\varepsilon$} (11);
      \backgroundExampleTAArrowDesc (20e) to[loop below] node[above, align=center]{$\varepsilon$} (20e);
      %
      \backgroundExampleTAArrowDesc (00) --node[above, align=center]{a} (10);
      \backgroundExampleTAArrowDesc (01) --node[above, align=center]{a} (11);
      %
      \backgroundExampleTAArrowDesc (10) --node[above, align=center]{b} (20);
      \backgroundExampleTAArrowDesc (11) --node[above, align=center]{b} (21);
      %
      \node[state, minimum size=0, inner sep=0, draw = none, above = 0.8*\prodSZGVertDistance of 20] (above20){};
      \draw (20) -- (above20){};
      \draw (20e) |- (above20){};
      \backgroundExampleTAArrowDesc (above20) -|node[above, align=center, xshift=5cm]{c} (00);
      %
      \backgroundExampleTAArrowDesc (21) --node[right, align=center, xshift=0.2cm]{sync} (20);
    \end{tikzpicture}
    }
    \caption{\PZVC{} of $A_1$ and $A_3$.}
    \label{fig:main:finding-a-timed-bisimulation:init-ex-pzvc}
  \end{figure}
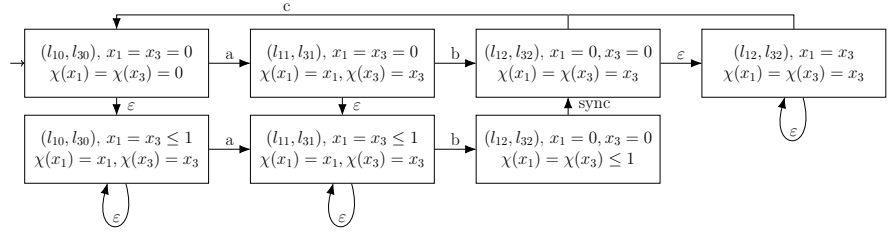
  Figure~\ref{fig:main:finding-a-timed-bisimulation:init-ex-pzvc} shows the \PZVC{} of $A_1$ and $A_3$. 
  In the initial state, all clocks are set to zero. 
  An $\varepsilon$-transition allows time to elapse as long as $x_1$ and $x_3$ are smaller or equal $1$. 
  Since the next $\varepsilon$-transition is a self-loop, we take the \texttt{a}-transition
  which is either followed by an $\varepsilon$-transition (self-loop), or the \texttt{b}-transition. 
  As both TA reset their clocks in this step, $x_1$ and $x_3$ are also reset in the \PZVC, while the virtual clocks remain unchanged. 
  Analogously to Example~\ref{ex:main:virtual-bisimulation:init-example}, we synchronize the symbolic state as the values of original and their virtual clocks may differ. 
  After another $\varepsilon$-transition, we reach $((l_{12}, l_{32}), x_1 = x_3 = \chi(x_1) = \chi(x_3))$. 
  Since $l_{12}$ and $l_{32}$ have outgoing \texttt{c}-transitions, this symbolic state also has one. 
  However, since the transition of $l_{32}$ has a guard, we must comply to it such that the target only contains the initial state again. 
  If $A_3$ would contain a reset instead of a guard (i.e., $A_3$ would be timed bisimilar to $A_1$), the zone of the target would be $x_1 = x_3 = 0 \land \chi(x_1) = \chi(x_3) \geq 0$ instead of the virtual clock values being equal to $0$.
\end{example}
%
%
\begin{definition}[\PZVC]
  \label{def:main:finding-a-timed-bisimulation:product-of-szg}
  A \PZVC{} of two TA $A$ and $B$ is a labeled transition system $(\SZGAllStates{}, s_{AB, 0}, \Sigma \allowbreak \cup \{\varepsilon\} \cup \{\text{sync}\}, \SZGTrans{})$. $\SZGAllStates{} \subseteq (L_A \times L_B) \times \mathcal{D}(C_{AB} \cup \chi(C_{AB}))$ is the set of all symbolic states $((l_A, l_B), \Zone{})$ with $\forall c \in C_{AB} : (\forall \ClockValuation{} \in \Zone{} : \ClockValuation{}(c) = \ClockValuation{}(\chi(c))) \lor (\forall \ClockValuation{} \in \Zone{} : \ClockValuation{}(c) = 0)$. We denote $\SZGAllStates{s} \subseteq \SZGAllStates{}$ for the set of all symbolic states $((l_A, l_B), \Zone{}) \in \SZGAllStates{}$ with $\forall c \in C_{AB} : \forall \ClockValuation{} \in \Zone{} : \ClockValuation{}(c) = \ClockValuation{}(\chi(c))$. Moreover, $s_{AB, 0} = ((l_{A, 0}, l_{B, 0}), [C_{AB} \cup \chi(C_{AB}) \rightarrow 0]) \in \SZGAllStates{s}$ is the initial symbolic state, $\Sigma \allowbreak \cup \{\varepsilon\} \cup \{\text{sync}\}$ is the set of transition labels, and $\SZGTrans{} \subseteq \SZGAllStates{} \times \Sigma \cup \{\varepsilon\} \cup \{\text{sync}\} \times \SZGAllStates{}$ is the least transition relation satisfying
  \begin{compactitem}
    \item $((l_A, l_B), \Zone{}) \SZGTrans{\text{sync}} ((l_A, l_B), \Zone{\text{sync}})$, if $((l_A, l_B), \Zone{}) \not\in \SZGAllStates{s}$, $\forall \ClockValuation{} \in \Zone{} : \exists \ClockValuation{\text{sync}} \in \Zone{\text{sync}} : \forall c \in C_{AB} : \ClockValuation{}(c) = \ClockValuation{\text{sync}}(c)$, $\forall \ClockValuation{\text{sync}} \in \Zone{\text{sync}} : \exists \ClockValuation{} \in \Zone{} : \forall c \in C_{AB} : \ClockValuation{}(c) = \ClockValuation{\text{sync}}(c)$, and $((l_A, l_B), \Zone{\text{sync}}) \in \SZGAllStates{s}$,
    \item $((l_A, l_B), \Zone{}) \SZGTrans{\varepsilon} ((l_A, l_B), \Zone{}^{\uparrow} \land I(l_A) \land I(l_B))$, if $((l_A, l_B), \Zone{}) \in \SZGAllStates{s}$, and
    \item $((l_{A, 1}, l_{B, 1}), \Zone{}) \SZGTrans{\sigma} ((l_{A, 2}, l_{B, 2}), R_A(R_B(\Zone{} \land g_A \land g_B)) \land I(l_{A, 2}) \land I(l_{B, 2}))$, if $((l_{A, 1}, l_{B, 1}), \Zone{}) \in \SZGAllStates{s}$, $l_{A, 1} \TATrans{g_A}{\sigma}{R_A} l_{A, 2}$, $l_{B, 1} \TATrans{g_B}{\sigma}{R_B} l_{B, 2}$, and $R_A(R_B(\Zone{} \land g_A \land g_B)) \land I(l_{A, 2}) \land I(l_{B, 2}) \neq \emptyset$.
  \end{compactitem}
\end{definition}

Again, we call a symbolic state $\SZGState{} \in \SZGAllStates{s}$ \emph{\synchronized}.
We next define symbolic virtual bisimulation over \PZVC.

\begin{definition}[Symbolic Virtual Bisimulation]
  \label{def:main:finding-a-timed-bisimulation:virtual-bisimulation-for-PZVC}
  A relation $R_{\text{symb}, \text{virt}} \subseteq $ \\ $\SZGAllStates{s}$ is a symbolic virtual bisimulation, if and only if $\text{states}(R_{\text{symb}, \text{virt}})$ is a virtual bisimulation according to Def.~\ref{def:main:virtual-bisimulation:virtual-bisimulation}.
\end{definition}


\begin{example}
\begin{figure}[tp]
    \centering
    \scalebox{\prodSZGScalingFactor}{
    \begin{tikzpicture}
        \tikzstyle{every node}=[font=\prodSZGTikzFontSize]
        \tikzstyle{symstate} = [draw,rectangle,minimum width=6cm,minimum height=2cm,inner sep=5pt,thick]
        %
        \node[symstate, align=center, initial left, initial text=] (00) {$(l_{10}, l_{20})$, $x_1=x_2=0$\\$\chi(x_1)=\chi(x_2)=0$};
        \node[symstate, align=center, right = \prodSZGHorDistance/2 of 00] (01) {$(l_{10}, l_{20})$, $x_1=x_2\leq1$\\$\chi(x_1)=x_1, \chi(x_2)=x_2$};
        \node[draw, rectangle, minimum width = 5.8cm,minimum height=1.8cm,inner sep=5pt,thick] (01double) at (01) {};
        %
        \node[symstate, align=center, right = \prodSZGHorDistance of 01] (10) {$(l_{11}, l_{21})$, $x_1=x_2\leq1$\\$\chi(x_1)=x_1, \chi(x_2)=x_2$};
        \node[draw, rectangle, minimum width = 5.8cm,minimum height=1.8cm,inner sep=5pt,thick] (10double) at (10) {};
        %
        \node[symstate, align=center, right = \prodSZGHorDistance/2 of 10] (21) {$(l_{12}, l_{22})$, $x_1=0, x_2\leq1$\\$\chi(x_1)=\chi(x_2)=x_2$};
        \node[symstate, align=center, below = \prodSZGVertDistance of 21] (20) {$(l_{12}, l_{22})$, $x_1=0, x_2\leq1$\\$\chi(x_1)=x_1, \chi(x_2)=x_2$};
        \node[symstate, align=center, left = \prodSZGHorDistance/2 of 20] (22) {$(l_{12}, l_{22})$, $0 \leq x_2 - x_1 \leq1$\\$\chi(x_1)=x_1, \chi(x_2)=x_2$};
        \node[draw, rectangle, minimum width = 5.8cm,minimum height=1.8cm,inner sep=5pt,thick] (22double) at (22) {};
        %
        \node[rectangle, minimum width=6cm,minimum height=2cm, align=center, left = \prodSZGHorDistance of 22] (30) {\dots};
        \node[symstate, align=center, left = \prodSZGHorDistance/2 of 30] (31) {$(l_{10}, l_{23})$, $x_1=x_2\leq1$\\$\chi(x_1)=x_1, \chi(x_2)=x_2$};
        \node[draw, rectangle, minimum width = 5.8cm,minimum height=1.8cm,inner sep=5pt,thick] (31double) at (31) {};
        %
        \node[symstate, align=center, below = \prodSZGVertDistance of 31] (40) {$(l_{11}, l_{24})$, $x_1=x_2\leq1$\\$\chi(x_1)=x_1, \chi(x_2)=x_2$};
        \node[draw, rectangle, minimum width = 5.8cm,minimum height=1.8cm,inner sep=5pt,thick] (40double) at (40) {};
        %
        \node[symstate, align=center, right = \prodSZGHorDistance/2 of 40] (50) {$(l_{12}, l_{22})$, $x_1=x_2=0$\\$\chi(x_1)=\chi(x_2)\leq1$};
        \node[symstate, align=center, right = \prodSZGHorDistance of 50] (51) {$(l_{12}, l_{22})$, $x_1=x_2=0$\\$\chi(x_1)=x_1, \chi(x_2)=x_2$};
        \node[symstate, align=center, right = \prodSZGHorDistance/2 of 51] (52) {$(l_{12}, l_{22})$, $x_1=x_2$\\$\chi(x_1)=x_1, \chi(x_2)=x_2$};
        %
        %
        \backgroundExampleTAArrowDesc (00) --node[above, align=center, yshift=0.2cm]{$\varepsilon$} (01);
        \backgroundExampleTAArrowDesc (20) --node[above, align=center, yshift=0.2cm]{$\varepsilon$} (22);
        \backgroundExampleTAArrowDesc (51) --node[above, align=center, yshift=0.2cm]{$\varepsilon$} (52);
        %
        %
        \backgroundExampleTAArrowDesc (01) --node[above, align=center, yshift=0.2cm]{a} (10);
        \backgroundExampleTAArrowDesc (31) --node[left, align=center, xshift=-0.2cm]{a} (40);
        %
        \backgroundExampleTAArrowDesc (10) --node[above, align=center, yshift=0.2cm]{b} (21);
        \backgroundExampleTAArrowDesc (40) --node[above, align=center, yshift=0.2cm]{b} (50);
        %
        \backgroundExampleTAArrowDesc (22) --node[above, align=center, yshift=0.2cm]{c} (30);
        %
        \backgroundExampleTAArrowDesc (21) --node[right, align=center, xshift=4]{sync} (20);
        \backgroundExampleTAArrowDesc (30) --node[above, align=center, yshift=0.2cm]{$\varepsilon$} (31);
        \backgroundExampleTAArrowDesc (50) --node[above, align=center, yshift=0.2cm]{sync} (51);
    \end{tikzpicture}
    }
  \caption{Finding a Witness for $A_1 \sim A_2$.}
  \label{fig:main:finding-a-timed-bisimulation:determ-ex}
\end{figure}
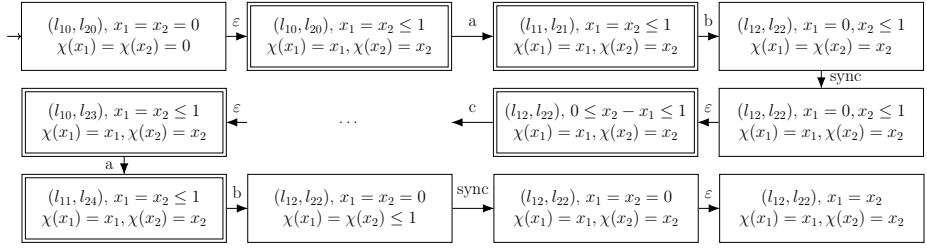
Figure~\ref{fig:main:finding-a-timed-bisimulation:determ-ex} shows how to create a witness for the bisimulation of $A_1$ and $A_2$ (Example~\ref{ex:background:TA}). We start at the initial symbolic state of the \PZVC{} and add the target of the $\varepsilon$-transition to the set (doubled border). 
The initial symbolic state is a subset of the target and, therefore, we remove the initial symbolic state and the set now only contains the second symbolic state.
This replacement is always feasible as 
the source of an epsilon transition is always a subset of its target.
Since the outgoing $\varepsilon$-transition of the second symbolic state is a self-loop, we ignore it. 
The target of the outgoing \texttt{a}-transition is added to the set. 
Since the second symbolic state does not have any other outgoing transitions, we now consider the newly added symbolic state. 
The $\varepsilon$-transition is a self-loop and, therefore, we consider the outgoing action transition. 
Since the target of the \texttt{b}-transition is not synchronized (as there are clock valuations in the zone for which $x_1$ and $\chi(x_1)$ are not equal), we add the symbolic state after the \texttt{sync}-transition, which is replaced by the target of the outgoing $\varepsilon$-transition. 
We continue until all outgoing transitions of all contained symbolic state have been processed. 
Finally, the symbolic virtual bisimulation contains those symbolic states marked with doubled border. 
The last symbolic state is a subset of the symbolic state after the second $\varepsilon$-transition and is, therefore, not added to the set.
\end{example}
In case of non-deterministic choices, this construction becomes more complicated. 
We call clock constraints over virtual clocks \emph{virtual constraints} 
and use \emph{contradictions}, which are sets of virtual constraints~\cite{FullReport}. 
The checking if a symbolic state can become part of 
a symbolic virtual bisimulation returns a contradiction.
If the symbolic state can be added, the contradiction is empty. 
Otherwise, the contradiction consists of virtual constraints such that each constraint can be satisfied by at least one state within the symbolic state. 
Any state that satisfies a virtual constraint of the contradiction cannot be part of a virtual bisimulation and must be excluded.
\begin{example}
\label{ex:main:finding-a-timed-bisimulation:virtual-bisimulation-nondeterm-case}
\newcommand{\findExampleAutomataScalingFactor}{0.45}
\newcommand{\findExampleAutomataSpaceBetween}{10mm}
\newcommand{\findExampleTAArrowDesc}{\draw[-{Latex[length=3mm]}]}
\newcommand{\findExampleTABendFactor}{80}
\newcommand{\findExTAVertDistance}{1cm}
\newcommand{\findExTAHorDistance}{1cm}
\newcommand{\findExampleTLTSArrowDesc}{\draw[-{Latex[length=2mm]}]}
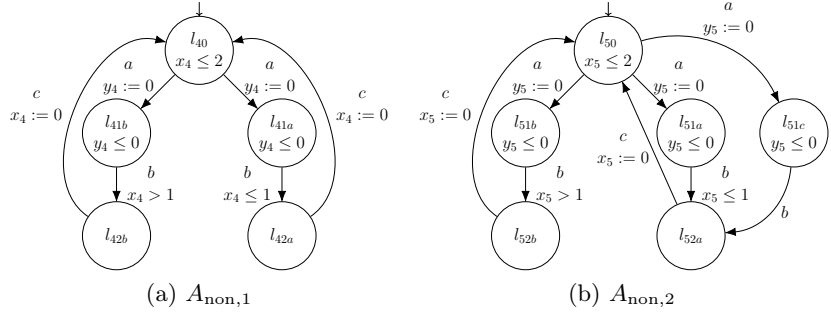
\begin{figure}[tp]
  \centering
  %
  \subfloat[$A_{\text{non}, 1}$\label{subfigure:main:finding-a-timed-bisimulation:nondeterm-example:TA-1-wit}]{
    \scalebox{\findExampleAutomataScalingFactor}{
    \begin{tikzpicture}
      \tikzstyle{every node}=[font=\prodSZGTikzFontSize]
      \tikzstyle{state} = [draw,circle,minimum size=2cm,inner sep=0pt,semithick]
      \node[state, align=center, initial above, initial text=] (0) {$l_{40}$\\$x_{\text{4}}\leq2$};
      \node[state, align=center, below right = \findExTAVertDistance and \findExTAHorDistance of 0] (1a) {$l_{41a}$\\$y_4 \leq 0$};
      \node[state, align=center, below = \findExTAVertDistance of 1a] (2a) {$l_{42a}$};
      \node[state, align=center, below left = \findExTAVertDistance and \findExTAHorDistance of 0] (1b) {$l_{41b}$\\$y_4 \leq 0$};
      \node[state, align=center, below = \findExTAVertDistance of 1b] (2b) {$l_{42b}$};
      \findExampleTAArrowDesc (0) --node[right, align=center, yshift=0.4cm]{$a$\\$y_4:=0$} (1a);
      \findExampleTAArrowDesc (1a) --node[left, align=center, xshift=-0.2cm]{$b$\\$x_4 \leq 1$} (2a);
      \findExampleTAArrowDesc (2a) edge [bend right = \findExampleTABendFactor] node[right, align=center, xshift=0.3cm] {$c$\\$x_{\text{4}}:=0$} (0);
      \findExampleTAArrowDesc (0) --node[left, align=center, yshift=0.4cm]{$a$\\$y_4:=0$} (1b);
      \findExampleTAArrowDesc (1b) --node[right, align=center, xshift=0.2cm]{$b$\\$x_4 > 1$} (2b);
      \findExampleTAArrowDesc (2b) edge [bend left = \findExampleTABendFactor] node[left, align=center, xshift=-0.3cm] {$c$\\$x_{\text{4}}:=0$} (0);
    \end{tikzpicture}
    }
  }%
  %
  %
  \subfloat[$A_{\text{non}, 2}$\label{subfigure:main:finding-a-timed-bisimulation:nondeterm-example:TA-2-wit}]{
    \scalebox{\findExampleAutomataScalingFactor}{
    \begin{tikzpicture}
      \tikzstyle{every node}=[font=\prodSZGTikzFontSize]
      \tikzstyle{state} = [draw,circle,minimum size=2cm,inner sep=0pt,semithick]
      \node[state, align=center, initial above, initial text=] (0) {$l_{50}$\\$x_{\text{5}}\leq2$};
      \node[state, align=center, below right = \findExTAVertDistance and \findExTAHorDistance of 0] (1a) {$l_{51a}$\\$y_5 \leq 0$};
      \node[state, align=center, below = \findExTAVertDistance of 1a] (2a) {$l_{52a}$};
      \node[state, align=center, below left = \findExTAVertDistance and \findExTAHorDistance of 0] (1b) {$l_{51b}$\\$y_5 \leq 0$};
      \node[state, align=center, below = \findExTAVertDistance of 1b] (2b) {$l_{52b}$};
      \node[state, align=center, right = \findExTAHorDistance of 1a] (1c) {$l_{51c}$\\$y_5 \leq 0$};
      \findExampleTAArrowDesc (0) --node[right, align=center, yshift=0.4cm]{$a$\\$y_5:=0$} (1a);
      \findExampleTAArrowDesc (1a) --node[right, align=center, xshift=0.2cm]{$b$\\$x_5 \leq 1$} (2a);
      \findExampleTAArrowDesc (2a) --node[left, align=center, xshift=0.3cm, xshift=-0.2cm, yshift=-0.2cm] {$c$\\$x_{\text{5}}:=0$} (0);
      \findExampleTAArrowDesc (0) --node[left, align=center, yshift=0.4cm]{$a$\\$y_5:=0$} (1b);
      \findExampleTAArrowDesc (1b) --node[right, align=center, xshift=0.2cm]{$b$\\$x_5 > 1$} (2b);
      \findExampleTAArrowDesc (2b) edge [bend left = \findExampleTABendFactor] node[left, align=center, xshift=-0.3cm] {$c$\\$x_{\text{5}}:=0$} (0);
      \findExampleTAArrowDesc (0) edge [bend left = 40] node[above, align=center, xshift=0.2cm, yshift=0.2cm]{$a$\\$y_5:=0$} (1c);
      \findExampleTAArrowDesc (1c) edge [bend left = 40] node[right, align=center, xshift=0.2cm]{$b$} (2a);
    \end{tikzpicture}
    }
  }%
\caption{Two timed bisimilar TA with Non-Deterministic Choices}
\label{fig:main::finding-a-timed-bisimulation::nondeterm-bisimilar-example}
\end{figure}
Figure~\ref{fig:main::finding-a-timed-bisimulation::nondeterm-bisimilar-example} shows two bisimilar TA with non-deterministic choices, where $A_{\text{non}, 2}$ has one additional branch. 
If action \texttt{a} is used while $x_5 \leq 1$ holds, the additional branch is timed bisimilar to the branch on the right hand side of $A_{\text{non}, 1}$. 
Therefore, we can combine these two branches into one symbolic state which leads to the contradiction $\{\chi(x_4) = \chi(y_4) = \chi(x_5) = \chi(y_5) > 1\}$. 
For $x_5 > 1$, the additional branch is timed bisimilar to the branch on the left hand side such that the contradiction $\{\chi(x_4) = \chi(y_4) = \chi(x_5) = \chi(y_5) \leq 1\}$ holds. 
\end{example}
%
%
To obtain finite witnesses, we utilize $k$-normalization~\cite{Rokicki,Pettersson}, where $k$ is the largest constant in any clock constraint of the TA thus serving as upper bound for all zones.
The image of normalization is finite and there exists a virtual bisimulation covering all states of a symbolic state if and only if one exists for its $k$-normalized counterpart~\cite{Cerans1992,FullReport}.
Since the initial symbolic state of the \PZVC{} contains the initial state of the \PTVC{}, any symbolic virtual bisimulation that includes this symbolic state also yields a virtual bisimulation containing the initial state of the \PTVC{}. 
To construct the bisimulation, we rely on forward- and backward-stability (see Appendix~\ref{app:fwd-and-bwd-stab}). 
Forward-stability guarantees that any transition from a state $((l_{A, 1}, l_{B, 1}), \ClockValuation{1}) \in ((l_{A, 1}, l_{B, 1}), \Zone{1})$ to a state $((l_{A, 2}, l_{B, 2}), \ClockValuation{2})$, whether due to action, delay, or synchronization, has a corresponding symbolic transition from $((l_{A, 1}, l_{B, 1}), \Zone{1})$ to $((l_{A, 2}, l_{B, 2}), \Zone{2})$ such that $((l_{A, 2}, l_{B, 2}), \ClockValuation{2}) \in ((l_{A, 2}, l_{B, 2}), \Zone{2})$. 
Backward-stability guarantees the reverse: If $((l_{A, 2}, l_{B, 2}), \Zone{2})$ has an incoming symbolic transition from $((l_{A, 1}, l_{B, 1}), \Zone{1})$, then every $((l_{A, 2}, l_{B, 2}), \ClockValuation{2}) \in ((l_{A, 2}, \allowbreak l_{B, 2}), \Zone{2})$ is reached by a transition from some $((l_{A, 1}, l_{B, 1}), \ClockValuation{1}) \in ((l_{A, 1}, l_{B, 1}), \Zone{1})$. 

\begin{proposition}
  \label{prop:main:finding-a-timed-bisimulation:virtual-bisimulation-for-PZVC}
  Given two timed bisimilar TA $A$ and $B$, a finite symbolic virtual bisimulation $R_{\text{symb}, \text{virt}}$ is constructed starting from the initial symbolic state by iteratively applying the following steps:
  \begin{compactitem}
    \item For any $((l_A, l_B), \Zone{}) \in R_{\text{symb}, \text{virt}}$ with outgoing $\varepsilon$-transition $((l_A, l_B), \Zone{}) \allowbreak \SZGTrans{\varepsilon} ((l_A, l_B), \Zone{\varepsilon})$, add the $k$-normalized $((l_A, l_B), \Zone{\varepsilon})$ to $R_{\text{symb}, \text{virt}}$ and
    \item for any $((l_A, l_B), \Zone{}) \in R_{\text{symb}, \text{virt}}$, if there exists a transition $l_A \TATrans{g_A}{\sigma}{R_A} l_{A, \sigma}$ with $\Zone{\sigma, A} = R_A(\Zone{} \land g_A) \land I(l_A) \neq \emptyset$, then there exists a finite set of symbolic states $\{((l_{A, \sigma}, \allowbreak l_{B, 1}), \allowbreak \Zone{\text{inter}, 1}), ..., \allowbreak ((l_{A, \sigma}, l_{B, n}), \Zone{\text{inter}, n})\}$ such that 
    \begin{compactenum}
      \item for any $((l_{A, \sigma}, l_{B, i}), \Zone{\text{inter}, i})$ there exists a symbolic state $((l_{A, \sigma}, l_{B, i}), \Zone{\sigma, i})$ with $\Zone{\text{inter}, i} \subseteq \Zone{\sigma, i}$ such that there exists a transition $((l_A, l_B), \Zone{}) \SZGTrans{\sigma} ((l_{A, \sigma}, l_{B, i}), \Zone{\sigma, i})$,
      \item if $((l_{A, \sigma}, l_{B, i}), \Zone{\text{inter}, i})$ is synchronized, it can be element of a symbolic virtual bisimulation,
      \item if $((l_{A, \sigma}, l_{B, i}), \Zone{\text{inter}, i})$ is not synchronized, there exists a transition $((l_{A, \sigma}, \allowbreak l_{B, i}), \Zone{\text{inter}, i}) \SZGTrans{\text{sync}} ((l_{A, \sigma}, l_{B, i}), \Zone{\text{sync}, i})$ and $((l_{A, \sigma}, l_{B, i}), \Zone{\text{sync}, i})$ can be element of a symbolic virtual bisimulation, and 
      \item for any $\ClockValuation{\sigma, A} \in \Zone{\sigma, A}$ there exists a $((l_{A, \sigma}, l_{B, i}), \allowbreak \Zone{\text{inter}, i})$ with a $\ClockValuation{\text{inter}, i} \models \phi_{\text{virt}}(\Zone{\text{inter}, i})$ with $\forall c \in C_A \cup \chi(C_{AB}) : \ClockValuation{\sigma, A}(c) = \ClockValuation{\text{inter}, i}(c)$ and vice versa.
    \end{compactenum}
    If $((l_{A, \sigma}, l_{B, i}), \Zone{\text{inter}, i})$ is synchronized, add its $k$-normalized version; if it is not synchronized, add its $k$-normalized version of $((l_{A, \sigma}, l_{B, i}), \Zone{\text{sync}, i})$, and
    \item do the analog for $B$.
  \end{compactitem}
\end{proposition}
The proof can be found in Appendix~\ref{app:create-virt-bisim}.
To summarize, for two bisimilar TA, $A$ and $B$, we first apply the algorithm from \cite{FullReport} and then apply the construction steps from Proposition~\ref{prop:main:finding-a-timed-bisimulation:virtual-bisimulation-for-PZVC}
to obtain a witness as shown in Example~\ref{ex:main:finding-a-timed-bisimulation:virtual-bisimulation-nondeterm-case}

\section{Finding Counterexamples}
\label{chap:main:finding-counterexamples}

We describe the construction of counterexamples showing that two TA are not timed bisimilar.
For deterministic TA, a counterexample represents a single finite path through the PVTC whose 
final state contradicts Def.~\ref{def:main:virtual-bisimulation:virtual-bisimulation}.

\begin{example}
\label{ex:main:new-def-virt-bisim:determ-contradiction-example}
To compare $A_1$ and $A_3$, we consider the sequence of transitions labeled \texttt{a}, \texttt{b}, \texttt{1} in the \PTVC{}.
The resulting state, $(l_{12}, l_{32}, x_1 = x_3 = \chi(x_1) = \chi(x_3) = 1)$, is synchronized, and there exists a transition $l_{12} \TATrans{true}{c}{\{x\}} l_{10}$. 
While the preconditions of the second case in Def.~\ref{def:main:virtual-bisimulation:virtual-bisimulation} are satisfied, the state lacks an outgoing transition labeled \texttt{c} according to Def.~\ref{def:main:virtual-bisimulation:product-synchronizable-tlts}. 
This is due to the fact that there is no outgoing transition from $l_{32}$ labeled \texttt{c} whose guard is fulfilled. 
\end{example}
However, in case a non-deterministic choice occurs, counterexamples become more complicated in terms of directed acyclic graphs (DAG).
\begin{example}
\label{ex:main:new-def-virt-bisim:nondeterm-contradiction-example}
\newcommand{\newdefExampleAutomataScalingFactor}{0.45}
\newcommand{\newdefExampleAutomataSpaceBetween}{10mm}
\newcommand{\newdefExampleTAArrowDesc}{\draw[-{Latex[length=3mm]}]}
\newcommand{\newdefExampleTABendFactor}{80}
\newcommand{\newdefExTAVertDistance}{1cm}
\newcommand{\newdefExTAHorDistance}{1cm}
\newcommand{\newdefExampleTLTSArrowDesc}{\draw[-{Latex[length=2mm]}]}
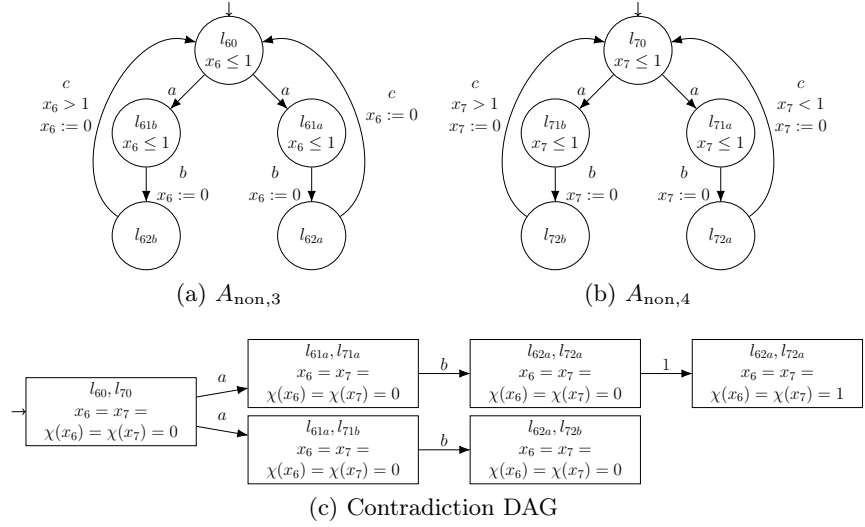
\begin{figure}[tp]
    \centering
    %
    \subfloat[$A_{\text{non}, 3}$\label{subfigure:main:new-def-virt-bisim:nondeterm-contradiction-example:TA-3-cont}]{
      \scalebox{\newdefExampleAutomataScalingFactor}{
      \begin{tikzpicture}
        \tikzstyle{every node}=[font=\newdefTikzFontSize]
        \tikzstyle{state} = [draw,circle,minimum size=2cm,inner sep=0pt,semithick]
        \node[state, align=center, initial above, initial text=] (0) {$l_{60}$\\$x_{\text{6}}\leq1$};
        \node[state, align=center, below right = \newdefExTAVertDistance and \newdefExTAHorDistance of 0] (1a) {$l_{61a}$\\$x_{\text{6}}\leq1$};
        \node[state, align=center, below = \newdefExTAVertDistance of 1a] (2a) {$l_{62a}$};
        \node[state, align=center, below left = \newdefExTAVertDistance and \newdefExTAHorDistance of 0] (1b) {$l_{61b}$\\$x_{\text{6}}\leq1$};
        \node[state, align=center, below = \newdefExTAVertDistance of 1b] (2b) {$l_{62b}$};
        \newdefExampleTAArrowDesc (0) --node[right, align=center, xshift=0.2cm]{$a$} (1a);
        \newdefExampleTAArrowDesc (1a) --node[left, align=center, xshift=-0.2cm]{$b$\\$x_{\text{6}}:=0$} (2a);
        \newdefExampleTAArrowDesc (2a) edge [bend right = \newdefExampleTABendFactor] node[right, align=center, xshift=0.3cm] {$c$\\$x_{\text{6}}:=0$} (0);
        \newdefExampleTAArrowDesc (0) --node[left, align=center, xshift=-0.2cm]{$a$} (1b);
        \newdefExampleTAArrowDesc (1b) --node[right, align=center, xshift=0.2cm]{$b$\\$x_{\text{6}}:=0$} (2b);
        \newdefExampleTAArrowDesc (2b) edge [bend left = \newdefExampleTABendFactor] node[left, align=center, xshift=-0.3cm] {$c$\\$x_{\text{6}} > 1$\\$x_{\text{6}}:=0$} (0);
      \end{tikzpicture}
      }
    }%
    %
    %
    \subfloat[$A_{\text{non}, 4}$\label{subfigure:main:new-def-virt-bisim:nondeterm-contradiction-example:TA-4-cont}]{
      \scalebox{\newdefExampleAutomataScalingFactor}{
      \begin{tikzpicture}
        \tikzstyle{every node}=[font=\newdefTikzFontSize]
        \tikzstyle{state} = [draw,circle,minimum size=2cm,inner sep=0pt,semithick]
        \node[state, align=center, initial above, initial text=] (0) {$l_{70}$\\$x_{\text{7}}\leq1$};
        \node[state, align=center, below right = \newdefExTAVertDistance and \newdefExTAHorDistance of 0] (1a) {$l_{71a}$\\$x_{\text{7}}\leq1$};
        \node[state, align=center, below = \newdefExTAVertDistance of 1a] (2a) {$l_{72a}$};
        \node[state, align=center, below left = \newdefExTAVertDistance and \newdefExTAHorDistance of 0] (1b) {$l_{71b}$\\$x_{\text{7}}\leq1$};
        \node[state, align=center, below = \newdefExTAVertDistance of 1b] (2b) {$l_{72b}$};
        \newdefExampleTAArrowDesc (0) --node[right, align=center, xshift=0.2cm]{$a$} (1a);
        \newdefExampleTAArrowDesc (1a) --node[left, align=center, xshift=-0.2cm]{$b$\\$x_{\text{7}}:=0$} (2a);
        \newdefExampleTAArrowDesc (2a) edge [bend right = \newdefExampleTABendFactor] node[right, align=center, xshift=0.3cm] {$c$\\$x_{\text{7}} < 1$\\$x_{\text{7}}:=0$} (0);
        \newdefExampleTAArrowDesc (0) --node[left, align=center, xshift=-0.2cm]{$a$} (1b);
        \newdefExampleTAArrowDesc (1b) --node[right, align=center, xshift=0.2cm]{$b$\\$x_{\text{7}}:=0$} (2b);
        \newdefExampleTAArrowDesc (2b) edge [bend left = \newdefExampleTABendFactor] node[left, align=center, xshift=-0.3cm] {$c$\\$x_{\text{7}} > 1$\\$x_{\text{7}}:=0$} (0);
      \end{tikzpicture}
      }
    }
    \\
    \subfloat[Contradiction DAG\label{subfigure:main:new-def-virt-bisim:nondeterm-contradiction-example:DAG}]{
      \scalebox{\newdefExampleAutomataScalingFactor}{
      \begin{tikzpicture}
        \tikzstyle{every node}=[font=\newdefTikzFontSize]
        \tikzstyle{state} = [draw,rectangle,minimum width=5cm, minimum height = 2cm, inner sep=0pt,semithick]
        \node[state, align=center, initial left, initial text=] (0) {$l_{60}, l_{70}$\\$x_{6}=x_{7}=$ \\ $\chi(x_6)=\chi(x_7)=0$};
        \node[state, align=center, above right = -0.9cm and 1.5cm of 0] (1a) {$l_{61a}, l_{71a}$\\$x_{6}=x_{7}=$ \\ $\chi(x_6)=\chi(x_7)=0$};
        \node[state, align=center, right = 1.5cm of 1a] (2a) {$l_{62a}, l_{72a}$\\$x_{6}=x_{7}=$ \\ $\chi(x_6)=\chi(x_7)=0$};
        \node[state, align=center, right = 1.5cm of 2a] (3a) {$l_{62a}, l_{72a}$\\$x_{6}=x_{7}=$ \\ $\chi(x_6)=\chi(x_7)=1$};
        \node[state, align=center, below right = -0.9cm and 1.5cm of 0] (1b) {$l_{61a}, l_{71b}$\\$x_{6}=x_{7}=$ \\ $\chi(x_6)=\chi(x_7)=0$};
        \node[state, align=center, right = 1.5cm of 1b] (2b) {$l_{62a}, l_{72b}$\\$x_{6}=x_{7}=$ \\ $\chi(x_6)=\chi(x_7)=0$};
        \newdefExampleTAArrowDesc (0) --node[above, align=center, yshift=0.2cm]{$a$} (1a);
        \newdefExampleTAArrowDesc (1a) --node[above, align=center]{$b$} (2a);
        \newdefExampleTAArrowDesc (2a) --node[above, align=center]{1} (3a);
        \newdefExampleTAArrowDesc (0) --node[above, align=center, yshift=0.2cm]{$a$} (1b);
        \newdefExampleTAArrowDesc (1b) --node[above, align=center]{$b$} (2b);
    \end{tikzpicture}
    }
  }
\caption{Two TA with Non-Deterministic Choices}
\label{fig:main:new-def-virt-bisim:nondeterm-contradiction-example}
\end{figure}
The top of Fig.~\ref{fig:main:new-def-virt-bisim:nondeterm-contradiction-example} shows two non-deterministic TA. 
Beside names of clocks and locations, the only difference 
is the additional guard in one of the transitions labeled \texttt{c} of $A_{\text{non}, 4}$. 
To check for timed bisimilarity, we assume a virtual bisimulation $R$ to contain the state $((l_{60}, l_{70}), x_6 = x_7 = \chi(x_6) = \chi(x_7) = 0)$ and consider the transition $l_{60} \TATrans{true}{a}{\emptyset} l_{61a}$. 
By Def.~\ref{def:main:virtual-bisimulation:virtual-bisimulation}, the existence of this transition implies either $((l_{61a}, l_{71a}), x_6 = x_7 = \chi(x_6) = \chi(x_7) = 0) \in R$ or $((l_{61a}, l_{71b}), x_6 = x_7 = \chi(x_6) = \chi(x_7) = 0) \in R$. 
Therefore, in case such a non-deterministic choice occurs, we have to consider all possible options, as depicted in Fig.~\ref{subfigure:main:new-def-virt-bisim:nondeterm-contradiction-example:DAG}. 
The first case implies $((l_{62a}, l_{72a}), x_6 = x_7 = \chi(x_6) = \chi(x_7) = 1) \in R$ and the second case implies $((l_{62a}, l_{72b}), x_6 = x_7 = \chi(x_6) = \chi(x_7) = 0) \in R$. 
In both cases, the part regarding $A_{\text{non}, 3}$ has an outgoing transition labeled \texttt{c}, while the part regarding $A_{\text{non}, 4}$ does not. 
Therefore, in both cases, $R$ is not a virtual bisimulation and $A_{\text{non}, 3}$ and $A_{\text{non}, 4}$ are not timed bisimilar.
\end{example}

The following definition generalizes the concept from Examples~\ref{ex:main:new-def-virt-bisim:determ-contradiction-example} and~\ref{ex:main:new-def-virt-bisim:nondeterm-contradiction-example}.

\begin{definition}[Contradiction DAGs]
\label{def:main:new-def-virt-bisim:DAG}
Let $A$, $B$ be two TA, with corresponding \PTVC{} $(\STLTSAllStates{}, s_{AB, 0}, \Sigma \cup \TimeDomain \cup \{\text{sync}\}, \STLTSTrans{})$. 
A \emph{Contradiction DAG} $(S, E)$ of $A$ and $B$ is a 
finite directed acyclic graph where $S \subseteq \STLTSAllStates{}$ 
is the set of states and $E \subseteq (S \times \Sigma \cup \TimeDomain \cup \{\text{sync}\} \times S)$ is the set of edges with $E \subseteq \STLTSTrans{}$. 
For any $s_1 \in S$, either $s_1 = ((l_{A, 1}, l_{B, 1}), \ClockValuation{1})$ 
is a leaf, in which case it is \synchronized{} and
\begin{compactitem}
\item either there exists a $d \in \TimeDomain$ such that $\ClockValuation{1} + d \models I(l_{A, 1})$ but no transition $s_1 \STLTSTrans{d} s_2$, or
\item there exists a transition $l_{A, 1} \TATrans{g_A}{\sigma}{R_A} l_{A, 2}$ with $\ClockValuation{1} \models g_A$ and $[R_A \rightarrow 0]\ClockValuation{1} \models I(l_{A, 2})$ but no transition $s_1 \STLTSTrans{\sigma} ((l_{A, 2}, l_{B, 2}), \ClockValuation{2})$ with $\forall c \in C_A : [R_A \rightarrow 0]\ClockValuation{1}(c) = \ClockValuation{2}(c)$, or
\item an analog condition for $B$
\end{compactitem}
holds, or it is not a leaf, in which case the sync-transition is element of $E$ in case $s_1$ is not synchronized, or there exists exactly one outgoing transition $(s_1, d, s_2) \in E$ with $d \in \TimeDomain$, or there exists a transition $l_{A, 1} \TATrans{g_A}{\sigma}{R_A} l_{A, 2}$ s.t.
\begin{compactitem}
\item for all edges $(s_1, \sigma_1, s_2) \in E$ the statement $\sigma_1 = \sigma$ holds, and
\item for any outgoing transition $l_{B, 1} \TATrans{g_B}{\sigma}{R_B} l_{B, 2}$, it holds that if there exists a transition $(s_1, \sigma, ((l_{A, 2}, l_{B, 2}), \ClockValuation{2})) \in \STLTSTrans{}$ with $\ClockValuation{2} = R_A(R_B(\ClockValuation{1}))$, then $(s_1, \sigma, \allowbreak ((l_{A, 2}, l_{B, 2}), \ClockValuation{2})) \in E$ holds,
\end{compactitem}
or the symmetric conditions hold, with the roles of $A$ and $B$ reversed.
\end{definition}


\begin{proposition}
  \label{prop:main:virtual-bisimulation:contradiction-DAGs-only-for-non-bisimilar}
   Let $A$, $B$ be two TA. $A$ and $B$ are virtually bisimilar if and only if the initial state of the \PTVC{} cannot be the root of a Contradiction DAG.
\end{proposition}
The proof can be found in Appendix~\ref{app:DAG-are-correct}.
A Contradiction DAG is finite thus no symbolic representation is required
and mapping to a Contradiction DAG without virtual clocks is also straightforward.
As shown in Example~\ref{ex:main:finding-a-timed-bisimulation:virtual-bisimulation-nondeterm-case}, the algorithm from~\cite{FullReport} produces contradictions for symbolic states that cannot belong to a symbolic virtual bisimulation. If a delay within a contradiction is possible, we use the highest possible delay that still fulfills the contradiction as it is unclear whether a particular state contributes to a contradiction solely due to a delay transition. Such a state may appear contradictory only because of the presence of that delay transition.
Once we reach a state where only delays less than 1 are allowed, we have to use a delay value s.t. any possible constraint of the TA does not change its evaluation if we get closer to the upper bound. In that case, we stay in the same \emph{region}, which is the base element for the decidability proof of \u{C}er\={a}ns~\cite{Cerans1992}. Two clock valuations $\ClockValuation{1}, \ClockValuation{2}$ are region-equivalent, if and only if
\begin{compactitem}
\item for any clock $c$, either $\ClockValuation{1}(c) > k \land \ClockValuation{2}(c) > k$ or the integral part of $\ClockValuation{1}(c)$ and $\ClockValuation{2}(c)$ are the same and the fractional part of $\ClockValuation{1}(c)$ is zero if and only if the fractional part of $\ClockValuation{2}(c)$ is zero and
\item for any clocks $c_1, c_2$, if the integral part of $\ClockValuation{1}(c_1)$ and $\ClockValuation{1}(c_2)$ are lower or equal $k$, the fractional part of $\ClockValuation{1}(c_1)$ is smaller than the fractional part of $\ClockValuation{1}(c_2)$ if and only if the the fractional part of $\ClockValuation{2}(c_1)$ is smaller than the fractional part of $\ClockValuation{2}(c_2)$.
\end{compactitem}
Let $\ClockValuation{}$ be the valuation and $0 < d_{\text{max}} < 1$ be the upper bound, we search for a value $0 \leq d_{\text{step}} < d_{\text{max}}$ s.t. for any delay $d$ between $d_{\text{step}}$ and $d_{\text{max}}$, $\ClockValuation{} + d$ is in the same region as $\ClockValuation{} + d_{\text{step}}$. If for any $c \in C$, $\ClockValuation{}(c)$ has a fractional value $\{\ClockValuation{}(c)\}$ such that $1 - \{\ClockValuation{}(c)\} < d_{\text{max}}$, we require $1 - \{\ClockValuation{}(c)\} < d_{\text{step}} < d_{\text{max}}$. Moreover, if for any $c \in C : \{\ClockValuation{}(c)\} = 0$, we require $0 < d_{\text{step}} < d_{\text{max}}$. In practice, we choose the center of the resulting interval.
In case there is no highest possible delay, we choose one such that all clock values are above the highest occurring constant in any of the TA.
The following proposition formalizes this approach.

\begin{proposition}
\label{prop:main:contradiction:find-a-DAG}
Let $A$ and $B$ be two non timed bisimilar TA with maximum used constant $k$.  
A Contradiction DAG can be constructed starting from the initial state as the root, by iteratively applying the following steps:
\begin{compactitem}
  \item If all leaves of the DAG serve as leaves of a Contradiction DAG, stop.
  \item If any leaf of the DAG is not synchronized, add a sync-transition.
  \item For any synchronized leaf $((l_A, l_B), \ClockValuation{})$ of the current DAG that cannot be a leaf of a Contradiction DAG, for which a $c$ exists such that $\ClockValuation{}(c) \leq k$, and for which a delay exists such that the target state fulfills the contradiction of its symbolic state and the resulting valuation is not region-equivalent to $\ClockValuation{}$, proceed as follows. If $\ClockValuation{} + k + 1$ fulfills the contradiction of the symbolic state, add the transition labeled $k + 1$, if the maximum delay satisfying the condition is an integer $d$, add the transition labeled $d$, and if the maximum lies strictly between two consecutive integers $d$ and $d+1$, add the transition labeled $d+d_{\text{step}}$.
  \item For any synchronized leaf $((l_A, l_B), \ClockValuation{})$ of the current DAG that cannot be a leaf of a Contradiction DAG and which does not fulfill the preconditions of the previous case, proceed as follows. There exists either a transition $l_A \TATrans{g_A}{\sigma}{R_A} l_{A,1}$ such that $\ClockValuation{} \models g_A$, $[R_A \rightarrow 0]\ClockValuation{} \models I(l_{A,1})$, and for any $l_B \TATrans{g_B}{\sigma}{R_B} l_{B,1}$ with $\ClockValuation{} \models g_B$ and $[R_B \rightarrow 0]\ClockValuation{} \models I(l_{B,1})$, the corresponding transition $((l_A, l_B), \ClockValuation{}) \STLTSTrans{\sigma} ((l_{A,1}, l_{B,1}), R_A(R_B(\ClockValuation{})))$ leads to a state whose synchronized clock valuation fulfills the contradiction of its symbolic state and we add all such transitions to the Contradiction DAG, or an analogous transition exists with the roles of $A$ and $B$ swapped.
  \item  If a transition introduces a cycle, there exists an alternative that allows the above procedure to continue without forming a cycle. Location-equivalent states with region-equivalent clock valuations are considered equivalent.
\end{compactitem}
\end{proposition}
The proof can be found in Appendix~\ref{app:create-cont-dag}.
\section{Experimental Evaluation}
\label{sec:evaluation}


\smallskip \noindent \emph{Implementation.} Our tool is open-source software\footnote{\url{https://github.com/Echtzeitsysteme/tchecker/}}. 
It extends the open-source model-checker TChecker\cite{TChecker}. 
We also provide a GUI~\footnote{\url{https://github.com/Echtzeitsysteme/tchecker-webapp}}.

\smallskip \noindent \emph{Research Questions.}
\begin{compactenum}
  \item[\textbf{(RQ1)}] What are the sizes of witnesses/counterexamples in comparison to the sizes of the input models?
  \item[\textbf{(RQ2)}] What is the runtime overhead caused by witness/counterexample generation compared to bare timed bisimulation checking?
\end{compactenum}

\smallskip \noindent \emph{Community Models.} We use three TA models 
widely used as community benchmarks for TA analysis techniques: \emph{Root Contention Protocol (RCP)}~\cite{IEEE1394RCP}, the \emph{IEEE1394 root contention protocol} (ten locations, 26 switches, two clocks); \emph{Audio/Video Components (AVC)} protocol for message transmissions over a shared bus (18 locations, 30 switches, single clock), which includes a non-deterministic choice, and the synchronous product of a variation of the \emph{train-gate model (TG)}~\cite{ALUR1994183} with three trains (73 locations, 129 switches, three clocks).

\smallskip \noindent \emph{Synthetic Models.} For RQ2, we also generated synthetic models. 
Our generator allows to adjust number of locations, model depth/width, and number of clocks. 

\smallskip \noindent \emph{Comparison Models.} 
For each community model, we generated three mutants using 
the approach in~\cite{MutationTesting} such that  
one mutant is timed bisimilar to the original while the others are not. 
This ground truth was calculated with our tool from~\cite{FullReport} and the witnesses/\allowbreak counterexamples were checked manually using the GUI.
Each original model was checked against its one-to-one copy and the mutants twice: once with and once without witness/counterexample generation. 

\smallskip \noindent \emph{Setup.} All runs were performed on an Intel i7-6700K processor with 64 GB of RAM, Linux Mint 21.2 ("Victoria"). All experiments were repeated ten times.

\begin{table}[tp]
  \centering
  \scalebox{.8}{
  \begin{tabular}{| c | c | c c c c c |}
    \hline
    & & \multicolumn{2}{c}{without} & \multicolumn{2}{c}{with} & Witness/DAG Size \\
    & & avg. time [ms] & variance & avg. time [ms] & variance & \# nodes \\
    \hline
    RCP & 1-to-1 & 1.5 & 0.0 & 2.0 & 0.2 & 10 \\
    & Bisim & 1.5 & 0.0 & 1.8 & 0.0 & 10 \\
    & Non-Bisim 1 & 0.4 & 0.0 & 1.0 & 0.3 & 5 \\
    & Non-Bisim 2 & 0.3 & 0.0 & 0.5 & 0.0 & 4 \\
    \hline
    AV & 1-to-1 & 0.9 & 0.0 & 1.1 & 0.0 & 20 \\
    & Bisim & 0.9 & 0.0 & 1.1 & 0.0 & 20 \\
    & Non-Bisim 1 & 0.9 & 0.0 & 2.2 & 1.0 & 21 \\
    & Non-Bisim 2 & 0.3 & 0.0 & 0.8 & 0.1 & 12 \\
    \hline
    TG & 1-to-1 & 228.4 & 0.3 & 376.4 & 1.0 & 765 \\
    & Bisim & 284.8 & 0.3 &  461.2 & 3.6 & 765 \\
    & Non-Bisim 1 & 233.4 & 3.1 & 237.9 & 0.2 & 18 \\
    & Non-Bisim 2 & 203096.6 & 10310.5 & 203211.1 & 53037.3 & 21 \\
    \hline
  \end{tabular}
  }
  \caption{Results for the community models}
  \label{tab:eval:results:CommunityModels}
\end{table}

\smallskip \noindent \emph{Results (Community Models).} 
The results are summarized in Table~\ref{tab:eval:results:CommunityModels}. 
For the RCP model, the highest average runtime is 2ms with low variances. The bisimilar cases produce witnesses whose size matches the number of TA locations. For the non-bisimilar cases, we have smaller Contradiction DAGs, while runtime is increased by a factor of up to 2.5.
For the AVC model, the sizes of the witness are, once again, equivalent to the number of locations, despite having the non-deterministic choice. 
We can see that witness generation noticeably impacts runtime, especially for non-bisimilar cases.
For the TG model, we see for the bisimilar cases that the average runtime is below 0.5s, still with low variances. However, the size of the witness is ten times the number of locations. In the first non-bisimilar case the average runtime is still quite low, while the Contradiction DAG is small. Since the second non-bisimilar case hits an edge case of the original algorithm, the average runtime increases to 3 minutes and 23 seconds.

\smallskip \noindent \emph{Results (Synthetic Models).} 
We used model sizes ranging from three up to 781 locations as well as models with four and five clocks.
Figure~\ref{fig:eval:eval_art_model} presents the evaluation results. 
The left-hand side shows the results for models with four clocks, while the right-hand side corresponds to models with five clocks. 
The first row depicts the evaluation duration with witness or DAG generation enabled, the second row shows the same evaluation without generation, and the last row illustrates the resulting runtime overhead in percent. Each graph includes results for models with depths of three and five, with the x-axis representing different values of the width parameter. There is always one graph for the experiments in which the corresponding model where checked against its 1-to-1 copy (\emph{timed bisimilar, tb}), and one graph in which the original model is checked against a version where the guard of the very last switch was modified (\emph{not tb}). 
For depth three, the number of locations increases from three to 31, while for depth five, it grows from five to 781. The corresponding runtimes range from below one millisecond up to approximately seven seconds. 

\begin{figure}[h!]
  \centering
  \subfloat[With Witness/DAG (4 Clocks)]{\includegraphics[width=0.4\textwidth]{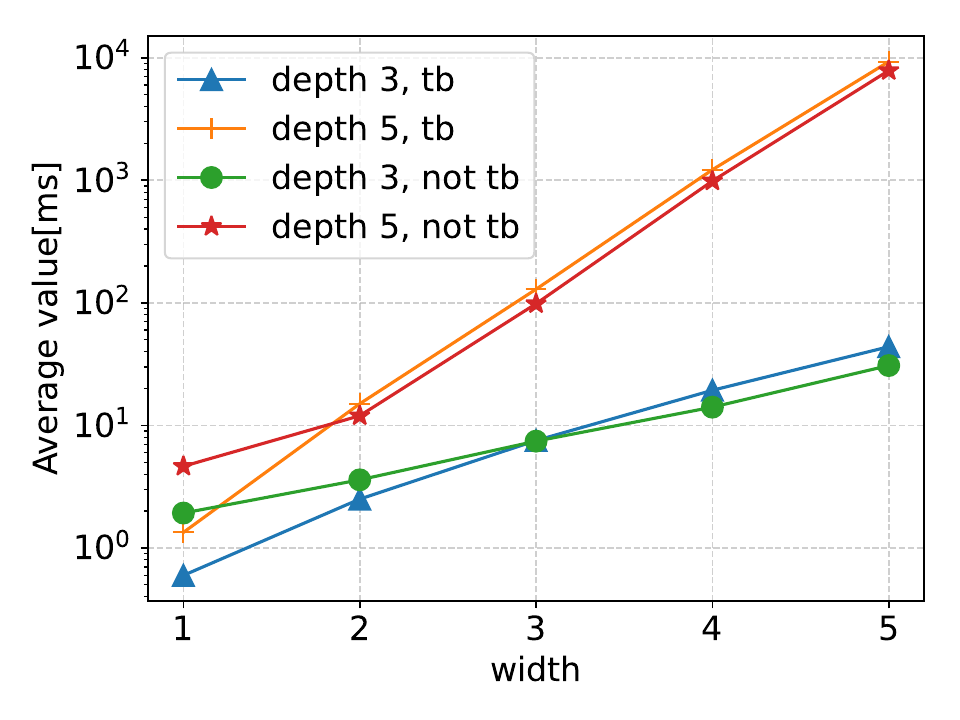}}\hfill
  \subfloat[With Witness/DAG (5 Clocks)]{\includegraphics[width=0.4\textwidth]{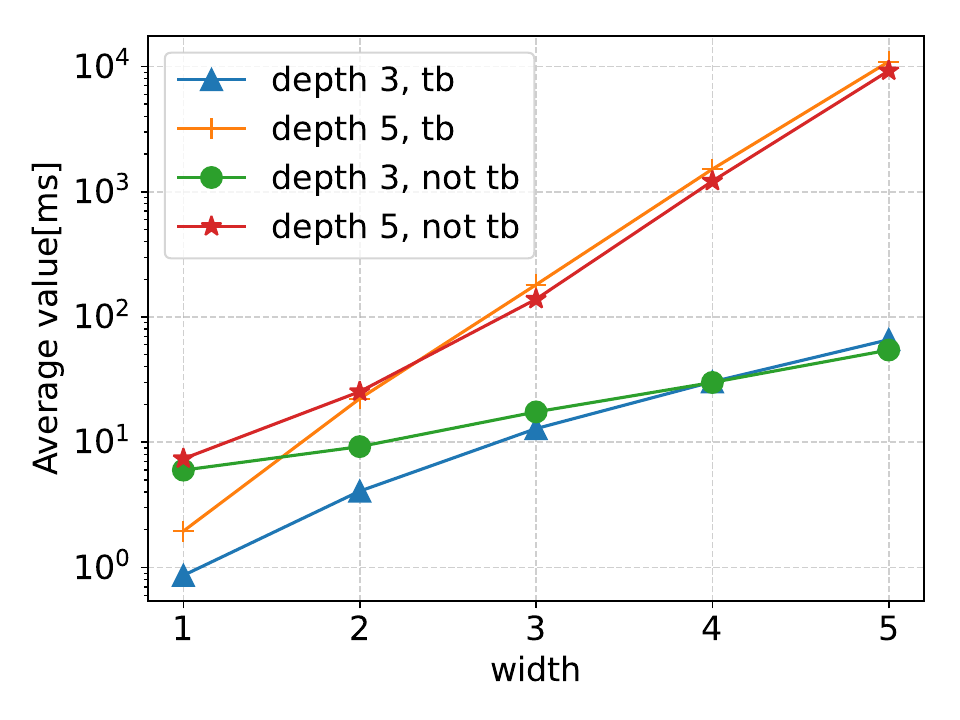}}\\[-1em]
  \subfloat[W/O Witness/DAG (4 Clocks)]{\includegraphics[width=0.4\textwidth]{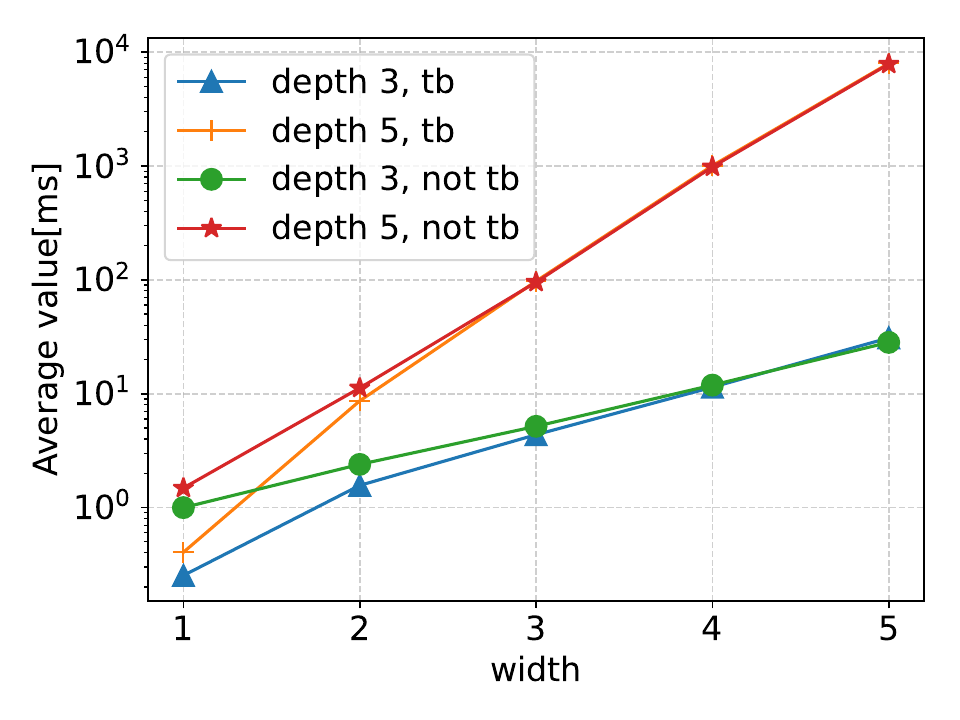}}\hfill
  \subfloat[W/O Witness/DAG (5 Clocks)]{\includegraphics[width=0.4\textwidth]{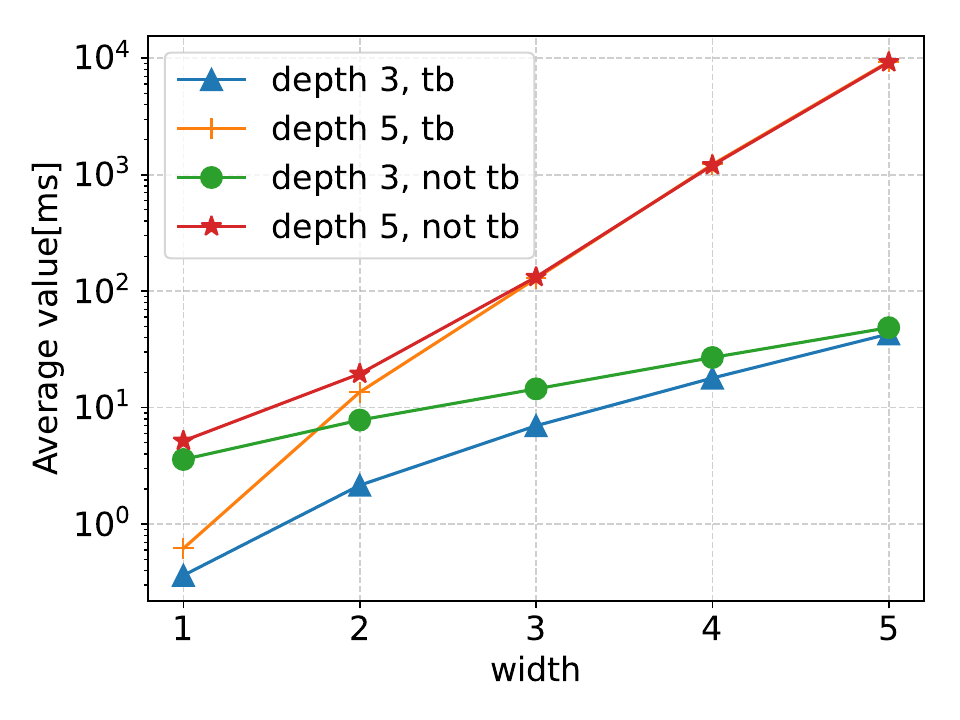}}\\[-1em]
  \subfloat[Runtime Overhead (4 Clocks)]{\includegraphics[width=0.4\textwidth]{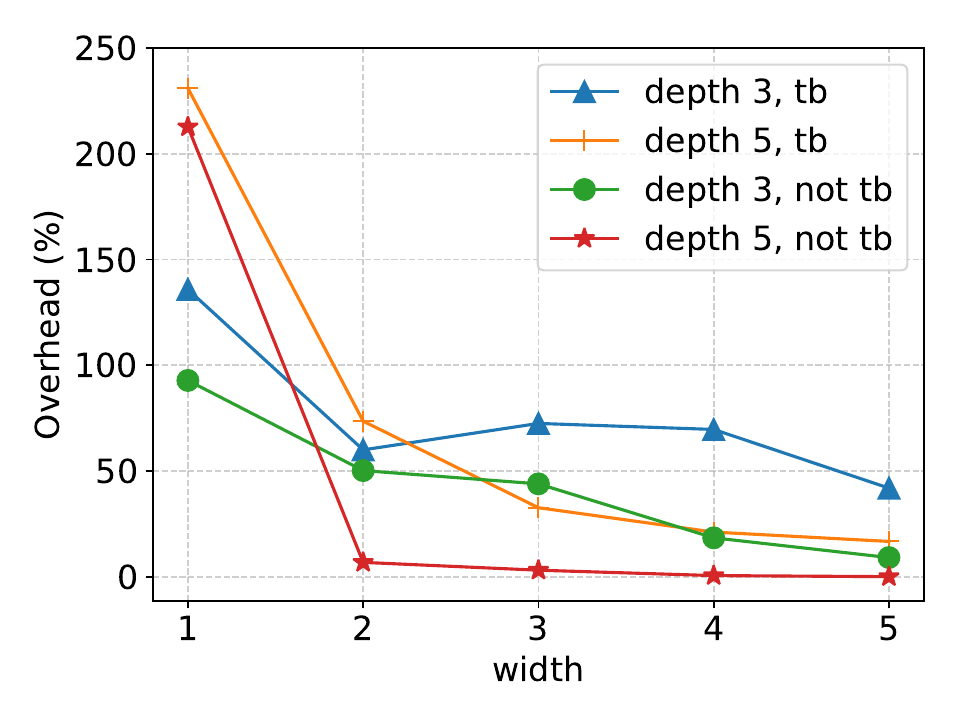}}\hfill
  \subfloat[Runtime Overhead (5 Clocks)]{\includegraphics[width=0.4\textwidth]{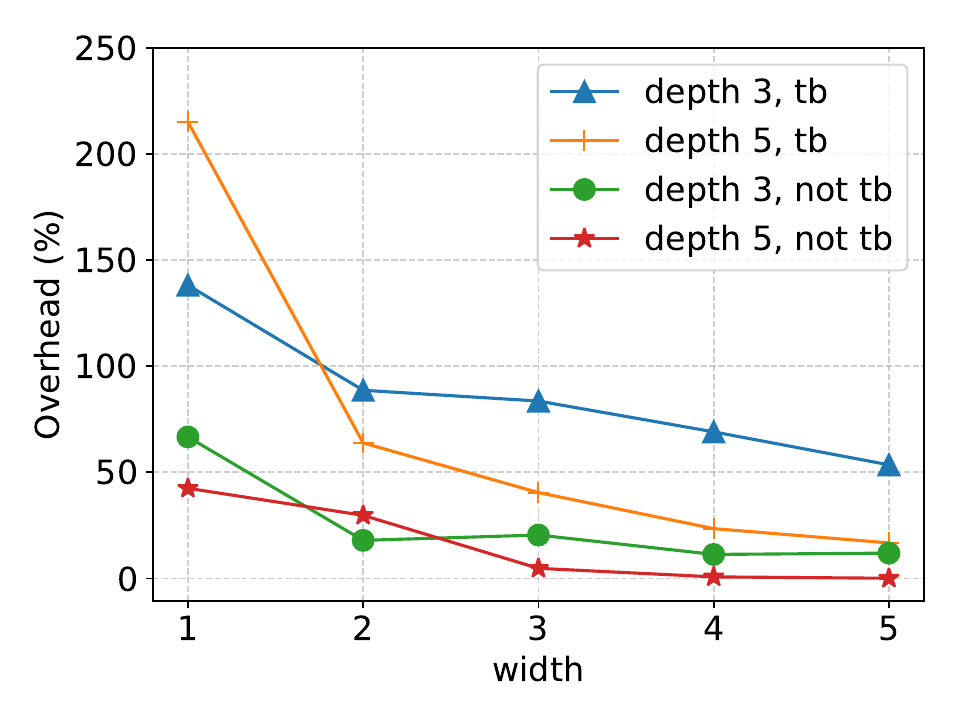}}
  \caption{Evaluation Results for the Artificial Model}
  \label{fig:eval:eval_art_model}
\end{figure}

The results reveal that, while the overhead is relatively large for smaller models, it decreases as the model size increases.

\smallskip \noindent \textbf{Answer to RQ1.} The number of locations is a possible indicator for the size of the Witness/Contradiction DAG, where
the Contradiction DAG is up to twice as large due to additional delay transitions. However, this indicator looses its reliability when the special structure of synchronous products comes into play.

\smallskip \noindent \textbf{Answer to RQ2.} The additional generation of the Witness/Contradiction DAG causes an acceptable runtime overhead
which decreases as model sizes increase.

\section{Conclusion}
\label{chap:conclusion}

We presented a method for generating witnesses to explain that two TA are timed bisimilar, as well as counterexamples to explain why they are not. Our open-source implementation extends the model checker TChecker, and our experimental evaluation shows that the approach is practically applicable and scalable.
As future work, we plan a case study to enhance our understanding of developers' needs and to assess the real-world applicability of our approach, for instance in automated test generation using mutation-based testing~\cite{MutationTesting,CortesEtAl}.
\FloatBarrier
\bibliographystyle{splncs04}
\bibliography{sections/references}

\begin{subappendices}
  \renewcommand{\thesection}{\Alph{section}}%
\section{Correctness of Virtual Bisimulation}%
\label{app:equi-def-vb}

In this appendix, we prove that two TA are virtual bisimilar according to Definition~\ref{def:main:virtual-bisimulation:virtual-bisimulation} if and only if they are timed bisimilar according to Definition~\ref{def:background:Strong-Timed-Bisimulation}. We begin by demonstrating that virtual bisimulation implies timed bisimulation and then establish the opposite direction.

\begin{proposition}
  \label{prop:app:equi-def-vb:VB-implies-TB}
  Given two virtual bisimilar TA $A$ and $B$. $A$ and $B$ are timed bisimilar.
\end{proposition}

\begin{proof}
To show timed bisimulation of $A$ and $B$, we have to show that there exists a timed bisimulation that contains the pair of initial states of their TLTS. Since $A$ and $B$ are virtual bisimilar, we know that there exists a virtual bisimulation $R_{\text{virt}}$, which contains the initial state of the \PTVC. We define the relation
\begin{align*}
  R_{\text{time}} = \{ & (\TLTSState{A}, \TLTSState{B}) \ | \ \exists ((l_A, l_B), \ClockValuation{AB}) \in R_{\text{virt}} : \\
  & \qquad (\forall c \in C_A : \ClockValuation{A}(c) = \ClockValuation{AB}(c) \land \forall c \in C_B : \ClockValuation{B}(c) = \ClockValuation{AB}(c)) \} 
\end{align*}
and show first that $R_{\text{time}}$ is a timed bisimulation and afterwards that the pair of initial states is element of $R_{\text{virt}}$.

To show that $R_{\text{time}}$ is a timed bisimulation, we have to show the conditions of Definition~\ref{def:background:Strong-Timed-Bisimulation}. Since the second can be shown analogously to the first, we only show the first. To show the first condition, we assume any pair $(\TLTSState{A, 1}, \TLTSState{B, 1}) \in R_{\text{time}}$ and any transition $\TLTSState{A, 1} \TLTSTrans{\mu} \TLTSState{A, 2}$ for $\mu \in \Sigma \cup \TimeDomain$, and show the existence of a transition $\TLTSState{B, 1} \TLTSTrans{\mu} \TLTSState{B, 2}$ with $(\TLTSState{A, 2}, \TLTSState{B, 2}) \in R_{\text{time}}$. By $(\TLTSState{A, 1}, \TLTSState{B, 1}) \in R_{\text{time}}$, we know that there exists an element $((l_{A, 1}, l_{B, 1}), \ClockValuation{\text{\PTVC}, 1}) \in R_{\text{virt}}$ with $\forall c \in C_A : \ClockValuation{A, 1}(c) = \ClockValuation{\text{\PTVC,} 1}(c)$ and $\forall c \in C_B : \ClockValuation{B, 1}(c) = \ClockValuation{\text{\PTVC,} 1}(c)$. Moreover, $((l_{A, 1}, l_{B, 1}), \ClockValuation{\text{\PTVC,} 1})$ is \synchronized. We distinguish two cases.
\begin{enumerate}
\item $\mu \in \TimeDomain$. In this case, $\ClockValuation{A, 1} + \mu \models I(l_{A, 1})$ and since $I(l_{A, 1}) \in \mathcal{B}(C_A)$ and $\forall c \in C_A : \ClockValuation{A, 1}(c) = \ClockValuation{\text{\PTVC,} 1}(c)$, $\ClockValuation{\text{\PTVC,} 1} + \mu \models I(l_{A, 1})$, which implies the existence of a transition $((l_{A, 1}, l_{B,1}), \ClockValuation{\text{\PTVC,} 1}) \STLTSTrans{\mu} ((l_{A, 1}, l_{B, 1}), \ClockValuation{\text{\PTVC,}  1} + \mu)$ with $((l_{A, 1}, l_{B, 1}), \ClockValuation{\text{\PTVC,}  1} + \mu) \in R_{\text{virt}}$ by Definition~\ref{def:main:virtual-bisimulation:virtual-bisimulation}. By Definition~\ref{def:main:virtual-bisimulation:product-synchronizable-tlts}, this implies $\ClockValuation{\text{\PTVC,} 1} + \mu \models I(l_{B, 1})$ and due to $I(l_{B, 1}) \in \mathcal{B}(C_B)$ and $\forall c \in C_B : \ClockValuation{B, 1}(c) = \ClockValuation{\text{\PTVC,} 1}(c)$, this implies the existence of a transition $\TLTSState{B, 1} \TLTSTrans{\mu} \TLTSState{B, 2}$ and since $\forall c \in C_A : \ClockValuation{A, 2}(c) = \ClockValuation{A, 1}(c) + \mu = \ClockValuation{\text{\PTVC,} 1}(c) + \mu = \ClockValuation{\text{\PTVC,} 2}(c)$ and the analog for $B$, $(\TLTSState{A, 2}, \TLTSState{B, 2}) \in R_{\text{time}}$ holds.
\item $\mu \in \Sigma$. In this case, there exists a transition $l_{A, 1} \TATrans{g_A}{\mu}{R_A} l_{A, 2}$ with $\ClockValuation{A, 1} \models g_A$, $\ClockValuation{A, 2} = [R_A \rightarrow 0]\ClockValuation{A, 1}$, $\ClockValuation{A, 2} \models I(l_{A, 2})$ by Definition~\ref{def:background:TLTS}. By Definition~\ref{def:main:virtual-bisimulation:virtual-bisimulation}, this implies the existence of a transition $((l_{A, 1}, l_{B,1}), \ClockValuation{\text{\PTVC,} 1}) \STLTSTrans{\mu} ((l_{A, 2}, l_{B, 2}), \ClockValuation{\text{\PTVC,}  2})$, which implies the existence of a transition $l_{B, 1} \TATrans{g_B}{\sigma}{R_B} l_{B, 2}$ with $\ClockValuation{\text{\PTVC,} 1} \models g_B$, $\ClockValuation{\text{\PTVC,} 2} = [R_A \cup R_B \rightarrow 0]\ClockValuation{\text{\PTVC,} 1}$, and $\ClockValuation{\text{\PTVC,} 2} \models I(l_{B, 2})$. This implies the existence of a transition $\TLTSState{B, 1} \TLTSTrans{\mu} \TLTSState{B, 2}$ with $\ClockValuation{B, 2} = [R_B \rightarrow 0]\ClockValuation{B, 1}$ by Definition~\ref{def:background:TLTS} and we only have to show that $(\TLTSState{A, 2}, \TLTSState{B, 2}) \in R_{\text{time}}$. We distinguish two cases
\begin{enumerate}
\item If $((l_{A, 2}, l_{B, 2}), \ClockValuation{\text{\PTVC,}  2})$ is synchronized, $((l_{A, 2}, l_{B, 2}), \ClockValuation{\text{\PTVC,}  2}) \in R_{\text{virt}}$ by Definition~\ref{def:main:virtual-bisimulation:virtual-bisimulation}. Since $\ClockValuation{\text{\PTVC,} 2} = [R_A \cup R_B \rightarrow 0] \ClockValuation{\text{\PTVC,} 1}$, $\forall c \in C_A : \ClockValuation{\text{\PTVC,} 1}(c) = \ClockValuation{A, 1}(c)$, and $\forall c \in C_B : \ClockValuation{\text{\PTVC,} 1}(c) = \ClockValuation{B, 1}(c)$, we know that $\forall c \in C_A : \ClockValuation{\text{\PTVC,} 2}(c) = [R_A \rightarrow 0]\ClockValuation{\text{\PTVC,} 1} = [R_A \rightarrow 0]\ClockValuation{A, 1}(c) = \ClockValuation{A, 2}(c)$ and the analog for $B$. Therefore, $(\TLTSState{A, 2}, \TLTSState{B, 2}) \in R_{\text{time}}$ by the definition of $R_{\text{time}}$.
\item If $((l_{A, 2}, l_{B, 2}), \ClockValuation{\text{\PTVC,}  2})$ is not synchronized, there exists a synchronized $((l_{A, 2}, l_{B, 2}), \ClockValuation{\text{\PTVC,}  2, s}) \in R_{\text{virt}}$ such that $\forall c \in C_A \cup C_B : \ClockValuation{\text{\PTVC,}  2, s}(c) = \ClockValuation{\text{\PTVC,}  2}(c)$ by Definition~\ref{def:main:virtual-bisimulation:product-synchronizable-tlts} and Definition~\ref{def:main:virtual-bisimulation:virtual-bisimulation}. Therefore, we can show $(\TLTSState{A, 2}, \TLTSState{B, 2}) \in R_{\text{time}}$ analogously to the previous case.
\end{enumerate}
\end{enumerate}
The initial state of the \PTVC{} is $((l_{0, A}, l_{0, B}), [C_A \cup C_B \cup \chi(C_A \cup C_B) \rightarrow 0])$, with $l_{0, A}$ being the initial location of $A$ and $l_{0, B}$ being the initial location of $B$. If this state is element of $R_{\text{virt}}$, the definition of $R_{\text{time}}$ implies $((l_{0, A}, [C \rightarrow 0]), (l_{0, B}, [C_B \rightarrow 0])) \in R_{\text{time}}$. Therefore, Proposition~\ref{prop:app:equi-def-vb:VB-implies-TB} holds.
\end{proof}

Since we now have seen that virtual bisimulation implies timed bisimulation, we now show that timed bisimulation implies virtual bisimulation.

\begin{proposition}
\label{prop:app:equi-def-vb:TB-implies-VB}
Given two timed bisimilar TA $A$ and $B$. $A$ and $B$ are virtually bisimilar.
\end{proposition}

\begin{proof}
To show virtual bisimulation of $A$ and $B$, we have to show that there exists a virtual bisimulation that contains the initial state of their \PTVC. Since $A$ and $B$ are timed bisimilar, we know that there exists a timed bisimulation $R_{\text{time}}$, which contains the initial states of the TLTS. We define the relation
\begin{align*}
  R_{\text{virt}} = \{ & ((l_{A, 1}, l_{B, 1}), \ClockValuation{AB}) \in \STLTSAllStates{AB} \ | \ ((l_{A, 1}, l_{B, 1}), \ClockValuation{AB}) \text{ is synchronized } \land \\
  & \exists (\TLTSState{A}, \TLTSState{B}) \in R_{\text{time}} : \\
  & \qquad (\forall c \in C_A : \ClockValuation{A}(c) = \ClockValuation{AB}(c) \land \forall c \in C_B : \ClockValuation{B}(c) = \ClockValuation{AB}(c)) \} 
\end{align*}
and show that $R_{\text{virt}}$ is a virtual bisimulation. To show that $R_{\text{virt}}$ is a virtual bisimulation, we consider any element $((l_{A, 1}, l_{B, 1}), \ClockValuation{\text{PTVC,} 1}) \in R_{\text{virt}}$ and show the conditions of Definition~\ref{def:main:virtual-bisimulation:virtual-bisimulation}. Since the conditions for $B$ can be shown analogously to the conditions for $A$, we only show the conditions for $A$.

By Definition of $R_{\text{virt}}$, there exists a pair of states $(\TLTSState[\text{TLTS}, A, 1]{A, 1}, \TLTSState[\text{TLTS}, B, 1]{B, 1}) \in R_{\text{time}}$ such that $\forall c \in C_A : \ClockValuation{\text{TLTS}, A, 1}(c) = \ClockValuation{\text{PTVC,} 1}(c)$ and $\forall c \in C_B : \ClockValuation{\text{TLTS}, B, 1}(c) = \ClockValuation{\text{PTVC,} 1}(c)$.

We start with the first condition. We assume any $d \in \TimeDomain$. In case $\ClockValuation{\text{PTVC,} 1} + d \not\models I(l_{A, 1})$, the condition holds trivially and, therefore, we assume $\ClockValuation{\text{PTVC,} 1} + d \models I(l_{A, 1})$ and have to show that there exists a transition $((l_{A, 1}, l_{B, 1}), \ClockValuation{\text{PTVC,} 1}) \STLTSTrans{d} ((l_{A, 1}, l_{B, 1}), \ClockValuation{\text{PTVC,} 1} + d)$ and $((l_{A, 1}, l_{B, 1}), \ClockValuation{\text{PTVC,} 1} + d) \in R_{\text{virt}}$.

Since $I(l_{A, 1}) \in \mathcal{B}(C_A)$, there exists a transition $\TLTSState[\text{TLTS}, A, 1]{A, 1} \TLTSTrans{d} (l_{A, 1}, \allowbreak \ClockValuation{\text{TLTS}, A, 1} + d)$. By Definition~\ref{def:background:Strong-Timed-Bisimulation}, this implies the existence of a transition $\TLTSState[\text{TLTS}, B, 1]{B} \TLTSTrans{d} (l_{B, 1}, \ClockValuation{\text{TLTS}, B, 1} + d)$ with $((l_{A, 1}, \ClockValuation{\text{TLTS}, A, 1} + d), (l_{B, 1}, \ClockValuation{\text{TLTS}, B, 1} + d)) \in R_{\text{time}}$. This can only be the case if $\ClockValuation{\text{TLTS}, B, 1} + d \models I(l_{B, 1})$. Since $I(l_{B, 1}) \in \mathcal{B}(C_B)$ and $\forall c \in C_B: (\ClockValuation{\text{TLTS}, B, 1} + d)(c) = (\ClockValuation{\text{PTVC,} 1} + d)(c)$, this implies $\ClockValuation{\text{PTVC,} 1} + d \models I(l_{B, 1})$. By Definition~\ref{def:main:virtual-bisimulation:product-synchronizable-tlts}, this implies the existence of a transition $((l_{A, 1}, l_{B, 1}), \ClockValuation{\text{PTVC,} 1}) \STLTSTrans{d} ((l_{A, 1}, l_{B, 1}), \ClockValuation{\text{PTVC,} 1} + d)$ and due to $\forall c \in C_A : (\ClockValuation{\text{PTVC,} 1} + d)(c) = (\ClockValuation{\text{TLTS}, A, 1} + d)(c)$ and the analog for $B$, $((l_{A, 1}, l_{B, 1}), \ClockValuation{\text{PTVC,} 1} + d) \in R_{\text{virt}}$ holds.

We now show the second condition. We assume a $\sigma \in \Sigma$ and transition $l_{A, 1} \TATrans{g_A}{\sigma}{R_A} l_{A, 2}$ with $\ClockValuation{\text{PTVC,} 1} \models I(g_A)$ and $[R_A \rightarrow 0]\ClockValuation{\text{PTVC,} 1} \models I(l_{A, 2})$ (otherwise, the condition holds trivially). We now have to show that there exists a transition $((l_{A, 1}, l_{B, 1}), \ClockValuation{\text{PTVC,} 1}) \STLTSTrans{\sigma} ((l_{A, 2}, l_{B, 2}), \ClockValuation{\text{PTVC,} 2})$ with $\forall c \in C_A : ([R_A \rightarrow 0 ]\ClockValuation{\text{PTVC,} 1})(c) = (\ClockValuation{\text{PTVC,} 2})(c)$ and either
\begin{itemize}
  \item $((l_{A, 2}, l_{B, 2}), \ClockValuation{\text{PTVC,} 2})$ is synchronized and $((l_{A, 2}, l_{B, 2}), \ClockValuation{\text{PTVC,} 2}) \in R_{\text{virt}}$ holds, or
  \item $((l_{A, 2}, l_{B, 2}), \ClockValuation{\text{PTVC,} 2})$ is not synchronized and there exists a transition $((l_{A, 2}, \allowbreak l_{B, 2}), \ClockValuation{\text{PTVC,} 2}) \STLTSTrans{\text{sync}} ((l_{A, 2}, l_{B, 2}), \ClockValuation{\text{PTVC,} 2, s})$ with $((l_{A, 2}, l_{B, 2}), \ClockValuation{\text{PTVC,} 2, s}) \in \\ R_{\text{virt}}$.
\end{itemize}

Since $\forall c \in C_A : \ClockValuation{\text{TLTS}, A, 1}(c) = \ClockValuation{\text{PTVC,} 1}(c)$, the transition $l_{A, 1} \TATrans{g_A}{\sigma}{R_A} l_{A, 2}$ implies by Definition~\ref{def:background:TLTS} the existence of a transition $\TLTSState[\text{TLTS}, A, 1]{A, 1} \TLTSTrans{\sigma} \TLTSState[\text{TLTS}, A, 2]{A, 2}$ with $\ClockValuation{\text{TLTS}, A, 2} = [R_A \rightarrow 0] \ClockValuation{\text{TLTS}, A, 1}$. By Definition~\ref{def:background:Strong-Timed-Bisimulation}, this implies the existence of a transition $\TLTSState[\text{TLTS}, B, 1]{B, 1} \TLTSTrans{\sigma} \TLTSState[\text{TLTS}, B, 2]{B, 2}$ with $(\TLTSState[\text{TLTS}, A, 2]{A, 2}, \TLTSState[\text{TLTS}, B, 2]{B, 2}) \in R_{\text{time}}$. By Definition~\ref{def:background:TLTS}, this can only be the case if there exists a transition $l_{B, 1} \TATrans{g_B}{\sigma}{R_B} l_{B, 2}$ with $\ClockValuation{\text{TLTS}, B, 1} \models g_B$, $\ClockValuation{\text{TLTS}, B, 2} = [R_B \rightarrow 0]\ClockValuation{\text{TLTS}, B, 1}$, and $\ClockValuation{\text{TLTS}, B, 2} \models I(l_{B, 2})$. Since $\forall c \in C_B : \ClockValuation{\text{PTVC,} 1}(c) = \ClockValuation{\text{TLTS}, B, 1}(c)$, this implies by Definition~\ref{def:main:virtual-bisimulation:product-synchronizable-tlts} a transition $((l_{A, 1}, l_{B, 1}), \ClockValuation{\text{PTVC,} 1}) \STLTSTrans{\sigma} ((l_{A, 2}, l_{B, 2}), \ClockValuation{\text{PTVC,} 2})$ with $\ClockValuation{\text{PTVC,} 2} = [R_A \cup R_B \rightarrow 0]\ClockValuation{\text{PTVC,} 1}$. Therefore, $\forall c \in C_A : \ClockValuation{\text{PTVC,}2}(c) = ([R_A \rightarrow 0] \ClockValuation{\text{PTVC,}1})(c) = ([R_A \rightarrow 0] \ClockValuation{\text{TLTS}, A, 1})(c) = \ClockValuation{\text{TLTS}, A, 2}(c)$ and the analog for $B$ holds. We distinguish two cases:
\begin{enumerate}
\item If $((l_{A, 2}, l_{B, 2}), \ClockValuation{\text{PTVC,} 2})$ is synchronized, then $((l_{A, 2}, l_{B, 2}), \ClockValuation{\text{PTVC,} 2}) \in R_{\text{virt}}$ holds by $(\TLTSState[\text{TLTS}, A, 2]{A, 2}, \TLTSState[\text{TLTS}, B, 2]{B, 2}) \in R_{\text{time}}$ and the definition of $R_{\text{virt}}$.
\item If $((l_{A, 2}, l_{B, 2}), \ClockValuation{\text{PTVC,} 2})$ is not synchronized, then we know by Definition~\ref{def:main:virtual-bisimulation:product-synchronizable-tlts} that there exists a transition $((l_{A, 2}, l_{B, 2}), \ClockValuation{\text{PTVC,} 2}) \STLTSTrans{\text{sync}} ((l_{A, 2}, l_{B, 2}), \ClockValuation{\text{PTVC,} 2, s})$ with $\forall c \in C_A \cup C_B : \ClockValuation{\text{PTVC,} 2}(c) = \ClockValuation{\text{PTVC,} 2, s}(c)$. Therefore, $((l_{A, 2}, l_{B, 2}), \allowbreak \ClockValuation{\text{PTVC,} 2, s}) \in R_{\text{virt}}$ holds by $(\TLTSState[\text{TLTS}, A, 2]{A, 2}, \TLTSState[\text{TLTS}, B, 2]{B, 2}) \in R_{\text{time}}$ and the definition of $R_{\text{virt}}$.
\end{enumerate}
If $((l_{0, A}, [C \rightarrow 0]), (l_{0, B}, [C_B \rightarrow 0])) \in R_{\text{time}}$, then $((l_{0, A}, l_{0, B}), [C_A \cup C_B \cup \chi(C_A \cup C_B) \rightarrow 0]) \in R_{\text{virt}}$ by definition of $R_{\text{virt}}$. Therefore, if the pair of initial states of the TLTS is element of $R_{\text{time}}$, then the initial state of the \PTVC{} is element of $R_{\text{virt}}$.
\end{proof}

Since for any two TA virtual bisimulation implies timed bisimulation and timed bisimulation implies virtual bisimulation, we have shown the equivalence of both definitions.
\section{Forward and Backward Stability}%
\label{app:fwd-and-bwd-stab}

In this appendix, we show that the \PTVC{} and the \PZVC{} are forward and backward stable.

\begin{proposition}[Forward Stability]
\label{prop:app:fwd-and-bwd-stab:forward-stability}
Given two TA $A$, $B$, with their \PTVC{} $(\STLTSAllStates{}, s_{\text{\PTVC}, 0}, \Sigma \cup \TimeDomain \cup \{\text{sync}\}, \STLTSTrans{})$ and their \PZVC{} $(\SZGAllStates{}, s_{\text{\PZVC}, 0}, \Sigma \cup \{\varepsilon\} \cup \{\text{sync}\}, \SZGTrans{})$, and a state $((l_A, l_B), \ClockValuation{}) \in \STLTSAllStates{}$.
\begin{enumerate}
\item If there exists a transition $((l_A, l_B), \ClockValuation{}) \allowbreak \STLTSTrans{\text{sync}} ((l_A, l_B), \ClockValuation{s})$ then for any $((l_A, l_B), \allowbreak \Zone{})$ with $((l_A, l_B), \ClockValuation{}) \in ((l_A, l_B), \Zone{})$ there exists a transition $((l_A, l_B), \Zone{}) \allowbreak \SZGTrans{\text{sync}} ((l_A, l_B), \Zone{s})$ with $((l_A, l_B), \ClockValuation{s}) \in ((l_A, l_B), \Zone{s})$. 
\item If there exists a transition $((l_A, l_B), \ClockValuation{}) \STLTSTrans{d} ((l_A, l_B), \ClockValuation{} + d)$ then for any synchronized symbolic state $((l_A, l_B), \Zone{})$ with $((l_A, l_B), \ClockValuation{}) \in ((l_A, l_B), \Zone{})$ there exists a transition $((l_A, l_B), \Zone{}) \SZGTrans{\varepsilon} ((l_A, l_B), \Zone{\varepsilon})$ with $((l_A, l_B), \ClockValuation{} + d) \in ((l_A, l_B), \Zone{\varepsilon})$.
\item If there exists a transition $((l_A, l_B), \ClockValuation{}) \STLTSTrans{\sigma} ((l_{A, 1}, l_{B, 1}), \ClockValuation{1})$ then for any synchronized symbolic state $((l_A, l_B), \Zone{})$ with $((l_A, l_B), \ClockValuation{}) \in ((l_A, l_B), \Zone{})$ there exists a transition $((l_A, l_B), \Zone{}) \SZGTrans{\sigma} ((l_{A, 1}, l_{B, 1}), \Zone{1})$ with $((l_{A, 1}, l_{B, 1}), \ClockValuation{1}) \allowbreak \in ((l_{A, 1}, l_{B, 1}), \Zone{1})$.
\end{enumerate}
\end{proposition}

\begin{proof}
We show one statement after the other.

In the first case, $((l_A, l_B), \ClockValuation{})$ is not synchronized, which implies that $((l_A, l_B), \allowbreak \Zone{})$ is not synchronized, which implies the existence of a transition $((l_A, l_B), \Zone{}) \allowbreak \SZGTrans{\text{sync}} ((l_A, l_B), \Zone{s})$ with $\forall \ClockValuation{} \in \Zone{} : \exists \ClockValuation{\text{sync}} \in \Zone{\text{s}} : \forall c \in C_{AB} : \ClockValuation{}(c) = \ClockValuation{\text{sync}}(c)$ and $((l_A, l_B), \Zone{\text{s}})$ is synchronized by Definition~\ref{def:main:finding-a-timed-bisimulation:product-of-szg} (which implies that all contained states are synchronized). By Definition~\ref{def:main:virtual-bisimulation:product-synchronizable-tlts}, $((l_A, l_B), \ClockValuation{s})$ is synchronized, which implies $\forall c \in C_{AB} : \ClockValuation{}(c) = \ClockValuation{}(\chi(c))$, and $\forall c \in C_{AB} : \ClockValuation{s}(c) = \ClockValuation{}(c)$ holds. Since the values of all clocks are specified, $\ClockValuation{s}$ is the only clock valuation of a synchronized state with the same original clock values as $\ClockValuation{}$. $((l_A, l_B), \ClockValuation{s}) \in ((l_A, l_B), \Zone{s})$ holds by $\forall \ClockValuation{} \in \Zone{} : \exists \ClockValuation{\text{sync}} \in \Zone{\text{s}} : \forall c \in C_{AB} : \ClockValuation{}(c) = \ClockValuation{\text{sync}}(c)$.

In the second case, we know by Definition~\ref{def:main:virtual-bisimulation:product-synchronizable-tlts} that $((l_A, l_B), \ClockValuation{})$ is synchronized and by the precondition that $((l_A, l_B), \Zone{})$ is synchronized. Moreover, we assume a transition $((l_A, l_B), \ClockValuation{}) \STLTSTrans{d} ((l_A, l_B), \ClockValuation{} + d)$ and we have to show that there exists a transition $((l_A, l_B), \Zone{}) \SZGTrans{\varepsilon} ((l_A, l_B), \Zone{\varepsilon})$ with $((l_A, l_B), \ClockValuation{} + d) \in ((l_A, l_B), \Zone{\varepsilon})$. By Definition~\ref{def:main:finding-a-timed-bisimulation:product-of-szg}, there exists a transition $((l_A, l_B), \Zone{}) \SZGTrans{\varepsilon} ((l_A, l_B), \Zone{}^{\uparrow} \land I(l_A) \land I(l_B))$. $\ClockValuation{} + d \in \Zone{}^{\uparrow}$ holds trivial by the definition of the future-operation (see Section~\ref{sec:background}). $(\ClockValuation{}+d) \models I(l_A)$ and $(\ClockValuation{} + d) \models I(l_B)$ hold by Definition~\ref{def:main:virtual-bisimulation:product-synchronizable-tlts}, which implies that $((l_A, l_B), \ClockValuation{} + d) \in ((l_A, l_B), \Zone{}^{\uparrow} \land I(l_A) \land I(l_B))$ holds.

In the third case, we know by Definition~\ref{def:main:virtual-bisimulation:product-synchronizable-tlts} that $((l_A, l_B), \ClockValuation{})$ is synchronized and by the precondition that $((l_A, l_B), \Zone{})$ is synchronized. Moreover, we assume a transition $((l_A, l_B), \ClockValuation{}) \STLTSTrans{\sigma} ((l_{A, 1}, l_{B, 1}), \ClockValuation{1})$ and we have to show that there exists a transition $((l_A, l_B), \Zone{}) \SZGTrans{\sigma} ((l_{A, 1}, l_{B, 1}), \Zone{1})$ with $((l_{A, 1}, l_{B, 1}), \ClockValuation{1}) \in ((l_{A, 1}, l_{B, 1}), \Zone{1})$. By Definition~\ref{def:main:virtual-bisimulation:product-synchronizable-tlts}, the given transition implies the existence of a transition $l_{A} \TATrans{g_A}{\sigma}{R_A} l_{A, 1}$ with $\ClockValuation{} \models g_A$ and the existence of a transition $l_{B} \TATrans{g_B}{\sigma}{R_B} l_{B, 1}$ with $\ClockValuation{} \models g_B$. Moreover, $\ClockValuation{1}(c) = [R_A \cup R_B \rightarrow 0] \ClockValuation{}(c)$, $\ClockValuation{1} \models I(l_{A, 1})$, and $\ClockValuation{1} \models I(l_{B, 1})$ hold. By Definition~\ref{def:main:finding-a-timed-bisimulation:product-of-szg}, these transitions imply the existence of a transition $((l_A, l_B), \Zone{}) \SZGTrans{\sigma} ((l_{A, 1}, l_{B, 1}), R_A(R_B(\Zone{} \land g_A \land g_B)) \land I(l_{A, 1}) \land I(l_{B, 1}))$, in case the target zone is not empty. Therefore, we only have to show that $\ClockValuation{1} \in R_A(R_B(\Zone{} \land g_A \land g_B)) \land I(l_{A, 1}) \land I(l_{B, 1})$. Since $\ClockValuation{}$ fulfills the guards, $\ClockValuation{} \in \Zone{} \land g_A \land g_B$ holds. Therefore, $\ClockValuation{1} \in R_A(R_B(\Zone{} \land g_A \land g_B))$ by $\ClockValuation{1}(c) = [R_A \cup R_B \rightarrow 0] \ClockValuation{}(c)$ and the definition of the reset-Operation. Due to $\ClockValuation{1} \models I(l_{A, 1})$, and $\ClockValuation{1} \models I(l_{B, 1})$, $\ClockValuation{1} \in R_A(R_B(\Zone{} \land g_A \land g_B)) \land I(l_{A, 1}) \land I(l_{B, 1})$ holds.

Therefore, the proposition holds.
\end{proof}

We now show backward stability.

\begin{proposition}[Backward Stability]
\label{prop:app:fwd-and-bwd-stab:backward-stability}
Given two TA $A$, $B$, with their \PTVC{} $(\STLTSAllStates{}, s_{\text{\PTVC}, 0}, \Sigma \cup \TimeDomain \cup \{\text{sync}\}, \STLTSTrans{})$ and their \PZVC{} $(\SZGAllStates{}, s_{\text{\PZVC}, 0}, \Sigma \cup \{\varepsilon\} \cup \{\text{sync}\}, \SZGTrans{})$, and a symbolic state $((l_A, l_B), \Zone{}) \in \SZGAllStates{}$.
\begin{enumerate}
\item If there exists a transition $((l_A, l_B), \Zone{}) \SZGTrans{\text{sync}} ((l_A, l_B), \Zone{s})$, for any $((l_A, l_B), \ClockValuation{s}) \allowbreak \in ((l_A, l_B), \Zone{s})$ either $((l_A, l_B), \ClockValuation{s}) \in ((l_A, l_B), \Zone{})$ holds or there exists a $((l_A, l_B), \ClockValuation{}) \in ((l_A, l_B), \Zone{})$ with $((l_A, l_B), \ClockValuation{}) \allowbreak \SZGTrans{\text{sync}} ((l_A, l_B), \ClockValuation{s})$.
\item If there exists a transition $((l_A, l_B), \Zone{}) \SZGTrans{\varepsilon} ((l_A, l_B), \Zone{\varepsilon})$, for any $((l_A, l_B), \ClockValuation{\varepsilon}) \allowbreak \in ((l_A, l_B), \Zone{\varepsilon})$, there exists a $((l_A, l_B), \ClockValuation{}) \in ((l_A, l_B), \Zone{})$ and a $d \in \TimeDomain$ such that $((l_A, l_B), \ClockValuation{}) \STLTSTrans{d} ((l_A, l_B), \ClockValuation{\varepsilon})$.
\item If there exists a transition $((l_A, l_B), \Zone{}) \SZGTrans{\sigma} ((l_{A, 1}, l_{B, 1}), \Zone{\sigma})$, for any $((l_{A, 1}, l_{B, 1} \allowbreak), \ClockValuation{\sigma}) \in ((l_{A, 1}, l_{B, 1}), \Zone{\sigma})$, there exists a $((l_A, l_B), \ClockValuation{}) \in ((l_A, l_B), \Zone{})$ such that $((l_A, l_B), \ClockValuation{}) \STLTSTrans{\sigma} ((l_{A, 1}, l_{B, 1}), \ClockValuation{\sigma})$.
\end{enumerate}
\end{proposition}

\begin{proof}
We show one statement after the other.

By Definition~\ref{def:main:finding-a-timed-bisimulation:product-of-szg}, there exists a $\ClockValuation{\text{prev}} \in \Zone{}$ such that $\forall c \in C_{AB} : \ClockValuation{\text{prev}}(c) = \ClockValuation{s}(c)$. We distinguish two cases:
\begin{enumerate}
\item If $((l_A, l_B), \ClockValuation{\text{prev}})$ is synchronized, we can follow:
\begin{equation*}
\forall \chi(c) \in \chi(C_{AB}) : \ClockValuation{\text{prev}}(\chi(c)) = \ClockValuation{\text{prev}}(c)
= \ClockValuation{s}(c)
= \ClockValuation{s}(\chi(c)).
\end{equation*}
Since the original and virtual clock values are equivalent, $\ClockValuation{\text{prev}} = \ClockValuation{s}$ holds and due to $\ClockValuation{\text{prev}} \in \Zone{}$ the statement $((l_A, l_B), \ClockValuation{s}) \in ((l_A, l_B), \Zone{})$ holds.
\item If $((l_A, l_B), \ClockValuation{\text{prev}})$ is not synchronized, we know by Definition~\ref{def:main:virtual-bisimulation:product-synchronizable-tlts} that there exists a transition $((l_A, l_B), \ClockValuation{\text{prev}}) \STLTSTrans{\text{sync}} ((l_A, l_B), \ClockValuation{\text{prev}, s})$ with $\forall c \in C_{AB} : \ClockValuation{\text{prev}}(c) = \ClockValuation{\text{prev}, s}(c)$ and $((l_A, l_B), \ClockValuation{\text{prev}, s})$ is synchronized. Since $\forall c \in C_{AB} : \ClockValuation{\text{prev}, s}(c) = \ClockValuation{\text{prev}}(c) = \ClockValuation{s}(c)$ and $\forall \chi(c) \in \chi(C_{AB}) : \ClockValuation{\text{prev}, s}(\chi(c)) = \ClockValuation{\text{prev}, s}(c) \allowbreak = \ClockValuation{s}(c) = \ClockValuation{s}(\chi(c))$, we can follow $\ClockValuation{\text{prev}, s} = \ClockValuation{s}$ and the statement holds.
\end{enumerate}

We now show the second statement. By Definition~\ref{def:main:finding-a-timed-bisimulation:product-of-szg}, we know that $((l_A, l_B), \allowbreak \Zone{})$ is synchronized and $\Zone{\varepsilon} = \Zone{}^{\uparrow} \land I(l_A) \land I(l_B)$. Therefore, for any $\ClockValuation{\varepsilon} \in \Zone{\varepsilon}$ exists a $\ClockValuation{} \in \Zone{}$ and a $d \in \TimeDomain$ such that $\ClockValuation{\varepsilon} = \ClockValuation{} + d$ by the definition of the future-Operator (see Section~\ref{sec:background}). Moreover, $((l_A, l_B), \ClockValuation{})$ is synchronized, since any element of a synchronized symbolic state is synchronized by Definition~\ref{def:main:finding-a-timed-bisimulation:product-of-szg}. Since $\ClockValuation{\varepsilon} \in \Zone{}^{\uparrow} \land I(l_A) \land I(l_B)$ implies $\ClockValuation{\varepsilon} \models I(l_A)$ and $\ClockValuation{\varepsilon} \models I(l_B)$, Definition~\ref{def:main:virtual-bisimulation:product-synchronizable-tlts} implies $((l_A, l_B), \ClockValuation{}) \STLTSTrans{d} ((l_A, l_B), \ClockValuation{\varepsilon})$ and the statement holds.

We now show the third statement. By Definition~\ref{def:main:finding-a-timed-bisimulation:product-of-szg}, $((l_{A}, l_{B}), \Zone{})$ is synchronized and $((l_{A}, l_{B}), \Zone{}) \SZGTrans{\sigma} ((l_{A, 1}, l_{B, 1}), \Zone{\sigma})$ implies the existence of a transition $l_{A} \TATrans{g_A}{\sigma}{R_A} l_{A, 1}$ and the analog for $B$ with $\Zone{\sigma} = R_A(R_B(\Zone{} \land g_A \land g_B)) \land I(l_{A, 1}) \land I(l_{B, 1})$. Since $\ClockValuation{\sigma} \in \Zone{\sigma}$, there exists a $\ClockValuation{} \in \Zone{} \land g_A \land g_B$ such that $\ClockValuation{\sigma} = [R_A \cup R_B \rightarrow 0]\ClockValuation{}$ by the definition of the reset-Operator (see Section~\ref{sec:background}). $\ClockValuation{\sigma} \in \Zone{\sigma}$ also implies $\ClockValuation{\sigma} \models I(l_{A, 1})$ and $\ClockValuation{\sigma} \models I(l_{B, 1})$, while $\ClockValuation{} \in \Zone{} \land g_A \land g_B$ implies $\ClockValuation{} \models g_A$ and $\ClockValuation{} \models g_B$ and $((l_A, l_B), \ClockValuation{}) \in ((l_{A}, l_{B}), \Zone{})$ is synchronized since $((l_{A}, l_{B}), \Zone{})$ is synchronized. Therefore, there exists a transition $((l_A, l_B), \ClockValuation{}) \STLTSTrans{\sigma} ((l_{A, 1}, l_{B, 1}), \ClockValuation{\sigma})$ by Definition~\ref{def:main:virtual-bisimulation:product-synchronizable-tlts}.
\end{proof}

\section{Construction of a Virtual Bisimulation}%
\label{app:create-virt-bisim}

We show Proposition~\ref{prop:main:finding-a-timed-bisimulation:virtual-bisimulation-for-PZVC} here. We start by an helping proposition, which shows that there exists such a set as described in the second step of Proposition~\ref{prop:main:finding-a-timed-bisimulation:virtual-bisimulation-for-PZVC} for symbolic states that can be element of a virtual bisimulation.

\begin{proposition}
\label{prop:app:create-virt-bisim:first-part-of-second-cond-for-virt-bisim-symb-states}
Given any symbolic state $((l_A, l_B), \Zone{})$ that can be element of a symbolic virtual bisimulation, if there exists a transition $l_A \TATrans{g_A}{\sigma}{R_A} l_{A, \sigma}$ with $\Zone{A, \sigma} = R_A(\Zone{} \land g_A) \land I(l_{A, \sigma}) \neq \emptyset$, then there exists a finite set of symbolic states $\{((l_{A, \sigma}, \allowbreak l_{B, 1}), \allowbreak \Zone{\text{inter}, 1}), ..., \allowbreak ((l_{A, \sigma}, l_{B, n}), \Zone{\text{inter}, n})\}$ such that 
    \begin{enumerate}
      \item for any $((l_{A, \sigma}, l_{B, i}), \Zone{\text{inter}, i})$ exists a symbolic state $((l_{A, \sigma}, l_{B, i}), \Zone{\sigma, i})$ with $\Zone{\text{inter}, i} \subseteq \Zone{\sigma, i}$ such that there exists a transition $((l_A, l_B), \Zone{}) \SZGTrans{\sigma} ((l_{A, \sigma}, l_{B, i}), \allowbreak \Zone{\sigma, i})$,
      \item if $((l_{A, \sigma}, l_{B, i}), \Zone{\text{inter}, i})$ is synchronized, it can be element of a symbolic virtual bisimulation,
      \item if $((l_{A, \sigma}, l_{B, i}), \Zone{\text{inter}, i})$ is not synchronized, there exists a transition $((l_{A, \sigma}, \allowbreak l_{B, i}), \Zone{\text{inter}, i}) \SZGTrans{\text{sync}} ((l_{A, \sigma}, l_{B, i}), \Zone{\text{sync}, i})$ and $((l_{A, \sigma}, l_{B, i}), \Zone{\text{sync}, i})$ can be element of a symbolic virtual bisimulation, and 
      \item for any $\ClockValuation{A, \sigma} \in \Zone{A, \sigma}$ there exists a $((l_{A, \sigma}, l_{B, i}), \allowbreak \Zone{\text{inter}, i})$ with a $\ClockValuation{\text{inter}, i} \models \phi_{\text{virt}}(\Zone{\text{inter}, i})$ with $\forall c \in C_A \cup \chi(C_{AB}) : \ClockValuation{A, \sigma}(c) = \ClockValuation{\text{inter}, i}(c)$ and vice versa.
    \end{enumerate}
\end{proposition}

\begin{proof}
We first show that there exists such a set and afterwards we show that it is finite. 

Since $((l_A, l_B), \Zone{})$ can be element of a symbolic virtual bisimulation, it must be synchronized. Therefore, any contained state can be element of a virtual bisimulation and must be synchronized, too. Since $\Zone{A, \sigma} \neq \emptyset$, there exists a $\ClockValuation{A, \sigma} \in \Zone{A, \sigma}$. By definition of $\Zone{A, \sigma}$, this can only be the case if $\ClockValuation{A, \sigma} \models I(l_{A, \sigma})$ and there exists a $\ClockValuation{} \in \Zone{} \land g_A$ such that $[R_A \rightarrow 0]\ClockValuation{} = \ClockValuation{A, \sigma}$. By Definition~\ref{def:main:virtual-bisimulation:virtual-bisimulation}, this implies the existence of a transition $((l_A, l_B), \ClockValuation{}) \STLTSTrans{\sigma} ((l_{A, \sigma}, l_{B, 1}), \ClockValuation{1})$ with $\forall c \in C_A \cup \chi(C_{AB}) : \ClockValuation{1}(c) = \ClockValuation{A, \sigma}(c)$ and either $((l_{A, \sigma}, l_{B, 1}), \ClockValuation{1})$ can be element of a virtual bisimulation (in case it is synchronized) or there exists an outgoing sync transition $((l_{A, \sigma}, l_{B, 1}), \ClockValuation{1}) \STLTSTrans{\text{sync}} ((l_{A, \sigma}, l_{B, 1}), \ClockValuation{\text{sync}})$ and $((l_{A, \sigma}, l_{B, 1}), \ClockValuation{\text{sync}})$ can be element of a virtual bisimulation. By forward stability (Proposition~\ref{prop:app:fwd-and-bwd-stab:forward-stability}), this implies the existence of a transition $((l_A, l_B), \Zone{}) \SZGTrans{\sigma} ((l_{A, \sigma}, l_{B, 1}), \Zone{1})$ with $\ClockValuation{1} \in \Zone{1}$. Let $\Zone{\text{inter}, 1} \subseteq \Zone{1}$ be the zone such that for all $\ClockValuation{\text{inter}} \in \Zone{\text{inter}, 1}$ exists a $\ClockValuation{\text{prev}} \in \Zone{}$ with $((l_A, l_B), \ClockValuation{\text{prev}}) \STLTSTrans{\sigma} ((l_{A, \sigma}, l_{B, 1}), \ClockValuation{\text{inter}})$ and they can be element of a virtual bisimulation and let $\{((l_{A, \sigma}, \allowbreak l_{B, 1}), \allowbreak \Zone{\text{inter}, 1}), ..., \allowbreak ((l_{A, \sigma}, l_{B, n}), \Zone{\text{inter}, n})\}$ be the set of all symbolic states with such zones. Then the first and last condition are trivially fulfilled and we only have to show the second and third condition.

In case $((l_{A, \sigma}, l_{B, i}), \Zone{\text{inter}, i})$ is synchronized, all contained states are synchronized and, therefore, all contained states can be element of a virtual bisimulation by construction. In case $((l_{A, \sigma}, l_{B, i}), \Zone{\text{inter}, i})$ is not synchronized, we know by backward stability (Proposition~\ref{prop:app:fwd-and-bwd-stab:backward-stability}) that all states of the target of the sync function can be element of a virtual bisimulation and, therefore, the set fulfills the conditions.

The set is finite, since we can divide $\Zone{A, \sigma}$ into a finite number of so-called regions \cite{AlurDill}. For any region, all contained states either fulfill or not fulfill the same conditions (remember that in clock constraints, only natural numbers are allowed). Therefore, the states of a region have the same outgoing transitions, which implies that all states of a region use the same transition and, therefore, there is no need to have more elements in the set as number of regions, which is finite (of course, the actual number of elements can be much smaller).
\end{proof}

The following proposition shows that any symbolic state contained in $R_{\text{symb}, \text{virt}}$ can be element of a symbolic virtual bisimulation.

\begin{proposition}
\label{prop:app:create-virt-bisim:states-r-contains-virt-bisim-only}
Given two timed bisimilar TA $A$ and $B$ and a set of symbolic states $R_{\text{symb}, \text{virt}}$, which is constructed as described in Proposition~\ref{prop:main:finding-a-timed-bisimulation:virtual-bisimulation-for-PZVC}. Any $((l_A, l_B), \Zone{}) \in R_{\text{symb}, \text{virt}}$ can be element of a symbolic virtual bisimulation.
\end{proposition}

\begin{proof}
We show this statement by induction in the number of symbolic states in $R_{\text{symb}, \text{virt}}$.\\%
\textbf{Hypothesis:} Any $((l_A, l_B), \Zone{}) \in R_{\text{symb}, \text{virt}}$ can be element of a symbolic virtual bisimulation.\\
\textbf{Base Case:} If $|R_{\text{symb}, \text{virt}}| = 1$, the hypothesis holds.\\
\textbf{Induction Step:} If the hypothesis holds for $|R_{\text{symb}, \text{virt}}| = n$, then the hypothesis also holds for $|R_{\text{symb}, \text{virt}}| = n$.

We first show the base case. Since we initially add the initial symbolic state, we know that in the base case the set only contains the initial symbolic state. The initial symbolic state only contains the initial state. Since $A$ and $B$ are timed bisimilar, they are also virtually bisimilar by Proposition~\ref{prop:app:equi-def-vb:TB-implies-VB}. This can only be the case, if there exists a virtual bisimulation that contains the initial state. Therefore, the initial state can be element of a virtual bisimulation, which implies that the initial symbolic state can be element of a symbolic virtual bisimulation.

We now show the induction step. We assume a $R_{\text{symb}, \text{virt}, n+1}$, which was created from a $R_{\text{symb}, \text{virt}, n}$ by applying the next step of Proposition~\ref{prop:main:finding-a-timed-bisimulation:virtual-bisimulation-for-PZVC}. By the induction hypothesis, all symbolic states of $R_{\text{symb}, \text{virt}, n}$ can be element of a virtual bisimulation and we have to show that the newly added symbolic state $((l_{A, n+1}, l_{B, n+1}), \Zone{n+1})$ can also be element of a virtual bisimulation, to show the induction step. We distinguish only two cases, as the cases for $B$ can be shown analogously.
\begin{itemize}
\item In this case, we assume that $((l_{A, n+1}, l_{B, n+1}), \Zone{n+1})$ was added due to the fact that there exists a $((l_{A, n}, l_{B, n}), \Zone{n}) \in R_{\text{symb}, \text{virt}, n}$ with $((l_{A, n}, l_{B, n}), \Zone{n}) \allowbreak \SZGTrans{\varepsilon} ((l_{A, n+1}, l_{B, n+1}), \Zone{n+1})$. By backward stability (Proposition~\ref{prop:app:fwd-and-bwd-stab:backward-stability}), we know that for any $((l_{A, n+1}, l_{B, n+1}), \ClockValuation{n+1}) \in ((l_{A, n+1}, l_{B, n+1}), \Zone{n+1})$, there exists a $\ClockValuation{n} \in \Zone{n}$ and a $d \in \TimeDomain$ such that $\ClockValuation{n+1} = \ClockValuation{n} + d$. Since $((l_{A, n}, l_{B, n}), \Zone{n})$ can be element of a symbolic virtual bisimulation, $((l_{A, n}, l_{B, n}), \ClockValuation{n})$ can be element of a virtual bisimulation, which implies that $((l_{A, n+1}, l_{B, n+1}), \ClockValuation{n+1})$ can also be element of a virtual bisimulation by Definition~\ref{def:main:virtual-bisimulation:virtual-bisimulation}. Since we decided on $((l_{A, n+1}, l_{B, n+1}), \ClockValuation{n+1})$ arbitrarily, we can follow that this holds for any contained state of $((l_{A, n+1}, l_{B, n+1}), \Zone{n+1})$, which implies that this symbolic state can be element of a symbolic virtual bisimulation.
\item In this case, we assume that $((l_{A, n+1}, l_{B, n+1}), \Zone{n+1})$ was added due to the second condition. Since it is explicitly mentioned that only symbolic states that can be element of a symbolic virtual bisimulation are added, the induction step holds in this case.
\end{itemize}
\end{proof}

We now show Proposition~\ref{prop:main:finding-a-timed-bisimulation:virtual-bisimulation-for-PZVC}.

\begin{proof}
Since we use k-normalized symbolic states, and the number of k-normalized symbolic states is finite \cite{Bengtsson2004}, the procedure described in Proposition~\ref{prop:main:finding-a-timed-bisimulation:virtual-bisimulation-for-PZVC} terminates and we only have to show that the result is a symbolic virtual bisimulation.

To show that $R_{\text{symb}, \text{virt}}$ is a  symbolic virtual bisimulation, we have to show that $\text{states}(R_{\text{symb}, \text{virt}})$ is a virtual bisimulation. To do so we assume any element $((l_A, l_B), \ClockValuation{}) \in \text{states}(R_{\text{symb}, \text{virt}})$ and show the conditions of Definition~\ref{def:main:virtual-bisimulation:virtual-bisimulation}. By definition of the $\text{states}$ function, there exists a $((l_A, l_B), \Zone{}) \in R_{\text{symb}, \text{virt}}$ such that $((l_A, l_B), \ClockValuation{}) \in ((l_A, l_B), \Zone{})$. By Proposition~\ref{prop:app:create-virt-bisim:states-r-contains-virt-bisim-only}, $((l_A, l_B), \Zone{})$ can be element of a symbolic virtual bisimulation. We only show the conditions for $A$, as the conditions for $B$ can be shown analogously.

To show the first condition, we assume a $d \in \TimeDomain$ such that $\ClockValuation{} + d \models I(l_A)$. We have to show that there exists a transition $((l_A, l_B), \ClockValuation{}) \STLTSTrans{d} ((l_A, l_B), \ClockValuation{} + d)$ and $((l_A, l_B), \ClockValuation{} + d) \in \text{states}(R_{\text{symb}, \text{virt}})$.
By Proposition~\ref{prop:app:create-virt-bisim:states-r-contains-virt-bisim-only}, we know that $((l_A, l_B), \ClockValuation{})$ can be element of a virtual bisimulation. Therefore, we know by Definition~\ref{def:main:virtual-bisimulation:virtual-bisimulation} that there exists a transition $((l_A, l_B), \ClockValuation{}) \STLTSTrans{d} ((l_A, l_B), \ClockValuation{} + d)$. By forward stability (Proposition~\ref{prop:app:fwd-and-bwd-stab:forward-stability}), there exists an outgoing $\varepsilon$-transition $((l_A, l_B), \Zone{}) \SZGTrans{\varepsilon} ((l_A, l_B), \Zone{\varepsilon})$ with $((l_A, l_B), \ClockValuation{} + d) \in ((l_A, l_B), \Zone{\varepsilon})$. Since the k-normalized version of $((l_A, l_B), \Zone{\varepsilon})$ is element of $R_{\text{symb}, \text{virt}}$, and a symbolic state is always a subset of its k-normalized form \cite{Bengtsson2004}, $((l_A, l_B), \ClockValuation{} + d) \in \text{states}(R_{\text{symb}, \text{virt}})$ holds and the first condition of Definition~\ref{def:main:virtual-bisimulation:virtual-bisimulation} is shown.

To show the second condition, we assume the existence of a transition $l_A \allowbreak \TATrans{g_A}{\sigma}{R_A} l_{A, \sigma}$ with $\ClockValuation{} \models g_A$ and $[R_A \rightarrow 0]\ClockValuation{} \models I(l_{A, \sigma})$ and have to show that there exists a transition $((l_A, l_B), \ClockValuation{}) \STLTSTrans{\sigma} ((l_{A, \sigma}, l_{B, 1}), \ClockValuation{1})$ with $\forall c \in C_A : [R_A \rightarrow 0]\ClockValuation{}(c) = \ClockValuation{1}(c)$ and either     
\begin{enumerate}
  \item $((l_{A, \sigma}, l_{B, 1}), \ClockValuation{1})$ is synchronized and $((l_{A, \sigma}, l_{B, 1}), \ClockValuation{1}) \in \text{states}(R_{\text{symb}, \text{virt}})$ holds, or
  \item $((l_{A, \sigma}, l_{B, 1}), \ClockValuation{1})$ is not synchronized and there exists a transition $((l_{A, \sigma}, \allowbreak l_{B, 1}), \allowbreak \ClockValuation{1}) \STLTSTrans{\text{sync}} ((l_{A, \sigma}, l_{B, 1}), \ClockValuation{1, s})$ with $((l_{A, \sigma}, l_{B, 1}), \ClockValuation{1, s}) \in \text{states}(R_{\text{symb}, \text{virt}})$.
\end{enumerate}
Since $((l_A, l_B), \ClockValuation{})$ can be element of virtual bisimulation by Proposition~\ref{prop:app:create-virt-bisim:states-r-contains-virt-bisim-only}, there exists a transition $((l_A, l_B), \ClockValuation{}) \STLTSTrans{\sigma} ((l_{A, \sigma}, l_{B, 1}), \ClockValuation{1})$ with $\forall c \in C_A : [R_A \rightarrow 0]\ClockValuation{}(c) = \ClockValuation{1}(c)$ and $((l_{A, \sigma}, l_{B, 1}), \ClockValuation{1})$ can be element of a virtual bisimulation and we only have to show the second part of the condition. By Proposition~\ref{prop:app:create-virt-bisim:first-part-of-second-cond-for-virt-bisim-symb-states}, there exists a set of symbolic states that fulfills the conditions of the second step of Proposition~\ref{prop:main:finding-a-timed-bisimulation:virtual-bisimulation-for-PZVC} such that it includes a symbolic state $(((l_{A, \sigma}, l_{B, 1}), \Zone{1}))$ such that $((l_{A, \sigma}, l_{B, 1}), \ClockValuation{1}) \in (((l_{A, \sigma}, l_{B, 1}), \Zone{1}))$. In case there exists more than one such outgoing transitions of $((l_A, l_B), \ClockValuation{})$, this statement holds for at least one of them. We assume that this is the set considered in the second step of Proposition~\ref{prop:main:finding-a-timed-bisimulation:virtual-bisimulation-for-PZVC}. We distinguish two cases:
\begin{enumerate}
\item If $((l_{A, \sigma}, l_{B, 1}), \Zone{1})$ is synchronized then it is added to $R_{\text{symb}, \text{virt}}$. In this case, $((l_{A, \sigma}, l_{B, 1}), \ClockValuation{1})$ is also synchronized and $((l_{A, \sigma}, l_{B, 1}), \ClockValuation{1}) \in \text{states}(R_{\text{symb}, \text{virt}})$ holds by $((l_{A, \sigma}, l_{B, 1}), \Zone{1}) \in R_{\text{symb}, \text{virt}}$.
\item If $((l_{A, \sigma}, l_{B, 1}), \Zone{1})$ is not synchronized, then there exists an outgoing transition $((l_{A, \sigma}, l_{B, 1}), \Zone{1}) \SZGTrans{\text{sync}} ((l_{A, \sigma}, l_{B, 1}), \Zone{1, \text{sync}})$ and the target is added to $R_{\text{symb}, \text{virt}}$. We distinguish two subcases.
\begin{itemize}
\item If $((l_{A, \sigma}, l_{B, 1}), \ClockValuation{1})$ is synchronized, then $((l_{A, \sigma}, l_{B, 1}), \ClockValuation{1}) \in ((l_{A, \sigma}, l_{B, 1}), \allowbreak \Zone{1, \text{sync}})$ by backward stability (Proposition~\ref{prop:app:fwd-and-bwd-stab:backward-stability}), which implies $((l_{A, \sigma}, l_{B, 1}), \allowbreak \ClockValuation{1}) \in \text{states}(R_{\text{symb}, \text{virt}})$ and the second condition of Definition~\ref{def:main:virtual-bisimulation:virtual-bisimulation} holds in this case.
\item If $((l_{A, \sigma}, l_{B, 1}), \ClockValuation{1})$ is not synchronized, there exists a transition $((l_{A, \sigma}, l_{B, 1}), \allowbreak \ClockValuation{1}) \STLTSTrans{\text{sync}} ((l_{A, \sigma}, l_{B, 1}), \ClockValuation{\text{sync}})$ with $((l_{A, \sigma}, l_{B, 1}), \ClockValuation{\text{sync}}) \in ((l_{A, \sigma}, l_{B, 1}), \Zone{1, \text{sync}})$ by backward stability (Proposition~\ref{prop:app:fwd-and-bwd-stab:backward-stability}), which implies $((l_{A, \sigma}, l_{B, 1}), \ClockValuation{1}) \in \text{states}(R_{\text{symb}, \text{virt}})$ and the second condition of Definition~\ref{def:main:virtual-bisimulation:virtual-bisimulation} holds in this case, too.
\end{itemize}
\end{enumerate}
Therefore, $\text{states}(R_{\text{symb}, \text{virt}})$ is a virtual bisimulation.
Since we add the initial symbolic state to $R_{\text{symb}, \text{virt}}$ and the initial state is element of the initial symbolic state, $\text{states}(R_{\text{symb}, \text{virt}})$ is a virtual bisimulation that contains the initial state and, therefore, a witness for the timed bisimulation of $A$ and $B$ by Definition~\ref{def:main:virtual-bisimulation:virtual-bisimulation}.
\end{proof}

\section{Contradiction DAGs are Correct}%
\label{app:DAG-are-correct}

In this appendix, we demonstrate that there exists a Contradiction DAG rooted at the initial state of the \PTVC{} if and only if the corresponding TA cannot be virtually bisimilar. We first show that if a Contradiction DAG exists, then the TA are not virtually bisimilar and afterwards the opposite direction.

\begin{proposition}
\label{prop:app:DAG-are-correct:if-ct-then-not-bisim}
If there exists a Contradiction DAG rooted at the initial state of the \PTVC, then the corresponding TA cannot be virtually bisimilar.
\end{proposition}

\begin{proof}
We show that for any state of a \PTVC{} that is element of a Contradiction DAG, there exists no virtual bisimulation according to Definition~\ref{def:main:virtual-bisimulation:virtual-bisimulation}, which contains this state. Since Contradiction DAGs are finite, there is also a finite maximum depth. We show the statement by induction in the maximum depth of the DAG that has the state under consideration as root. \\
\textbf{Hypothesis: } Any state $s_1$ that is element of a Contradiction DAG cannot be element of a virtual bisimulation. \\
\textbf{Base Case: } Any leaf of a Contradiction DAG cannot be element of a virtual bisimulation.\\
\textbf{Induction Step: } If the hypothesis holds if the element has at most $n$ steps to any leaf, then it also holds if the element has at most $n+1$ steps to any leaf.

We first show the base case. By Definition~\ref{def:main:new-def-virt-bisim:DAG}, a leaf has to fulfill one of several conditions. We only analyze the first two conditions, as the analog conditions for $B$ can be shown analogously.
\begin{itemize}
\item If there exists a $d \in \TimeDomain$ such that $\ClockValuation{1} + d \models I(l_{A, 1})$ but no transition $s_1 \STLTSTrans{d} s_2$, then $s_1$ cannot fulfill the first condition of Definition~\ref{def:main:virtual-bisimulation:virtual-bisimulation}.
\item If there exists a transition $l_{A, 1} \TATrans{g_A}{\sigma}{R_A} l_{A, 2}$ with $\ClockValuation{1} \models g_A$ and $[R_A \rightarrow 0]\ClockValuation{1} \models I(l_{A, 2})$ but no transition $s_1 \STLTSTrans{\sigma} ((l_{A, 2}, l_{B, 2}), \ClockValuation{2})$ with $\forall c \in C_A : [R_A \rightarrow 0]\ClockValuation{1}(c) = \ClockValuation{2}(c)$, then $s_1$ cannot fulfill the second condition of Definition~\ref{def:main:virtual-bisimulation:virtual-bisimulation}.
\end{itemize}
Therefore, for any relation $R$ with $s_1 \in R$, one of the conditions of Definition~\ref{def:main:virtual-bisimulation:virtual-bisimulation} is not fulfilled and $R$ is not a virtual bisimulation. We now show the induction step.

Since $s_1$ is not a leaf in the induction step, we know that there exists at least a single directed edge $(s_1, \sigma, s_2) \in E$. We distinguish three cases:
\begin{enumerate}
\item $\sigma = \text{sync}$. In this case, $s_1$ is not synchronized and cannot be element of a virtual bisimulation by Definition~\ref{def:main:virtual-bisimulation:virtual-bisimulation}. 
\item $\sigma \in \TimeDomain$. In this case, $s_1$ is synchronized and has exactly a single outgoing transition labeled with $\sigma$ by Definition~\ref{def:main:virtual-bisimulation:product-synchronizable-tlts}. Since we know by the induction hypothesis that there exists no virtual bisimulation that contains $s_2$, Definition~\ref{def:main:virtual-bisimulation:virtual-bisimulation} implies that there exists no virtual bisimulation that contains $s_1$.
\item $\sigma \in \Sigma$. We denote $s_1 = ((l_{A, 1}, l_{B, 1}), \ClockValuation{1})$ and $s_2 = ((l_{A, 2}, l_{B, 2}), \ClockValuation{2})$. We assume that there exists a transition $l_{A, 1} \TATrans{g_A}{\sigma}{R_A} l_{A, 2}$ such that
\begin{itemize}
\item for all edges $(s_1, \sigma_1, s_2) \in E$ the statement $\sigma_1 = \sigma$ holds,
\item for any outgoing transition $(s_1, \sigma, ((l_{A, 2}, l_{B, 2}), \ClockValuation{2})) \in E$, there exists a transition $l_{B, 1} \TATrans{g_B}{\sigma}{R_B} l_{B, 2}$ such that $\ClockValuation{2} = R_A(R_B(\ClockValuation{1}))$,
\item for any outgoing transition $l_{B, 1} \TATrans{g_B}{\sigma}{R_B} l_{B, 2}$ it holds that if there exists a transition $(s_1, \sigma, ((l_{A, 2}, l_{B, 2}), \ClockValuation{2})) \in \STLTSTrans{}$ with $\ClockValuation{2} = R_A(R_B(\ClockValuation{1}))$, then $(s_1, \sigma, \allowbreak ((l_{A, 2}, l_{B, 2}), \ClockValuation{2})) \in E$ holds,
\end{itemize}
since in case there exists such a transition $l_{B, 1} \TATrans{g_B}{\sigma}{R_B} l_{B, 2}$, the induction step can be shown analogously. The second condition of Definition~\ref{def:main:virtual-bisimulation:virtual-bisimulation} describes two conditions for $s_1$, caused by the existence of the transition $l_{A, 1} \TATrans{g_A}{\sigma}{R_A} l_{A, 2}$. We show that none of those conditions are fulfilled by any $(s_1, \sigma_1, s_2) \in E$.
\begin{enumerate}
\item If $s_2$ is synchronized, we can apply the induction hypothesis and, therefore, $s_2$ cannot be element of a virtual bisimulation.
\item If $s_2$ is not synchronized, it has an outgoing sync transition $s_2 \STLTSTrans{\text{sync}} ((l_{A, 2}, l_{B, 2}), \ClockValuation{2, s})$, which is also element of the Contradiction DAG. The target has at most $n$ steps (actually, it has $n-1$) to a leaf and, therefore, cannot be element of a virtual bisimulation by the induction hypothesis.
\end{enumerate}
However, by Definition~\ref{def:main:virtual-bisimulation:product-synchronizable-tlts}, there can be no other outgoing transition of $s_1$ in $\STLTSTrans{}$, labeled $\sigma$, that fulfills the condition $\forall c \in C_A : [R_A \rightarrow 0]\ClockValuation{1}(c) = \ClockValuation{2}(c)$. Therefore, the second condition of Definition~\ref{def:main:virtual-bisimulation:virtual-bisimulation} cannot be fulfilled by $s_1$ and $s_1$ cannot be element of a virtual bisimulation.
\end{enumerate}
We have shown that no state that is element of a Contradiction DAG can be element of a virtual bisimulation. Since a Contradiction DAG is always rooted at the initial state, there can be no virtual bisimulation that contains this state in case there exists a Contradiction DAG.
\end{proof}

We now show the opposite direction. To do so, we introduce bounded virtual bisimulation, analogously to \cite{FullReport}.

\begin{definition}[Bounded Virtual Bisimulation]
  \label{def:app:DAG-are-correct:bounded-virtual-bisimulation}
  Given two TA $A$, $B$, with their \PTVC{} $(\STLTSAllStates{}, s_{AB, 0}, \Sigma \cup \TimeDomain \cup \{\text{sync}\}, \STLTSTrans{})$. Any relation $R_{\text{virt}, 0} \subseteq \STLTSAllStates{}$ that only contains synchronized states is a virtual bisimulation in order 0. A relation  $R_{\text{virt}, n+1} \subseteq \STLTSAllStates{}$ for $n \in \mathbb{N}^{\geq 0}$ that only contains synchronized states is a virtual bisimulation in order n+1 if and only if it holds that
  \begin{enumerate}
    \item for any $d \in \TimeDomain$ with $\ClockValuation{} + d \models I(l_A)$, there exists a transition $((l_A, l_B), \ClockValuation{}) \STLTSTrans{d} ((l_A, l_B), \ClockValuation{} + d)$ and there exists a virtual bisimulation $R_{\text{virt}, n}$ in order n such that $((l_A, l_B), \ClockValuation{} + d) \in R_{\text{virt}, n}$,
    \item if there exists a transition $l_A \TATrans{g_A}{\sigma}{R_A} l_{A, 1}$ with $\ClockValuation{} \models g_A$ and $[R_A \rightarrow 0] \ClockValuation{} \models I(l_{A, 1})$, then there exists a transition $((l_A, l_B), \ClockValuation{}) \STLTSTrans{\sigma} ((l_{A, 1}, l_{B, 1}), \ClockValuation{1})$ with $\forall c \in C_A : [R_A \rightarrow 0] \ClockValuation{}(c) = \ClockValuation{1}(c)$ and either
    \begin{enumerate}
      \item $((l_{A, 1}, l_{B, 1}), \ClockValuation{1})$ is synchronized and there exists a virtual bisimulation $R_{\text{virt}, n}$ in order n such that $((l_{A, 1}, l_{B, 1}), \ClockValuation{1}) \in R_{\text{virt}, n}$ holds, or
      \item $((l_{A, 1}, l_{B, 1}), \ClockValuation{1})$ is not synchronized and there exists a transition $((l_{A, 1}, \allowbreak l_{B, 1}), \ClockValuation{1}) \STLTSTrans{\text{sync}} ((l_{A, 1}, l_{B, 1}), \ClockValuation{1, s})$ there exists a virtual bisimulation $R_{\text{virt}, n}$ in order n such that $((l_{A, 1}, l_{B, 1}), \ClockValuation{1, s}) \in R_{\text{virt}}$, and
    \end{enumerate}
    \item the analog for $B$.
  \end{enumerate}
\end{definition}

For $n \rightarrow \infty$, $R_{\text{virt}, n}$ is a virtual bisimulation. If we define a series of functions $f_n$, that maps $\STLTSAllStates{}$ to the subset $R_n \subseteq \STLTSAllStates{}$ that contains all states of $\STLTSAllStates{}$ that can be element of a virtual bisimulation in order $n$, then there exists a finite $m$ such that for all $n > m$ the statement $R_{\text{virt}, m} = R_{\text{virt}, n}$ holds. For more information see \cite{FullReport}. Using this definition, we can show the following proposition.

\begin{proposition}
\label{prop:app:DAG-are-correct:if-not-bisim-then-ct}
If the TA are not virtually bisimilar then there exists a Contradiction DAG rooted at the initial state of the corresponding \PTVC.
\end{proposition}

\begin{proof}
We show by induction that any synchronized state that cannot be element of a virtual bisimulation in order n, can be the root of a Contradiction DAG with maximum depth 2n. \\
\textbf{Hypothesis: } Any synchronized state that cannot be element of a virtual bisimulation in order n can be the root of a Contradiction DAG with maximum depth 2n. \\
\textbf{Base Case: } The hypothesis holds for n=0 and n=1. \\
\textbf{Induction Step: } If the hypothesis holds for any value lower than n, then it also holds for n+1.

Since there exists no synchronized state that cannot be element of a virtual bisimulation in order 0, the hypothesis holds trivially for n=0. Since any synchronized state can be element of a virtual bisimulation in order 0, we know by Definition~\ref{def:app:DAG-are-correct:bounded-virtual-bisimulation} that for a synchronized state $((l_A, l_B), \ClockValuation{})$ that cannot be element of a virtual bisimulation in order 1, one of the following statements hold:
\begin{enumerate}
\item Either there exists a $d \in \TimeDomain$ such that $\ClockValuation{} + d \models I(l_A)$ but no transition $((l_A, l_B), \ClockValuation{}) \STLTSTrans{d} ((l_A, l_B), \ClockValuation{} + d)$ or
\item there exists a transition $l_A \TATrans{g_A}{\sigma}{R_A} l_{A, 1}$ with $\ClockValuation{} \models g_A$ and $[R_A \rightarrow 0]\ClockValuation{} \models I(l_{A, 1})$ but no transition $((l_A, l_B), \ClockValuation{}) \STLTSTrans{\sigma} ((l_{A, 1}, l_{B, 1}), \ClockValuation{1})$ with $\forall c \in C_A : [R_A \rightarrow 0]\ClockValuation{}(c) = \ClockValuation{1}(c)$ or
\item any analog condition for $B$ hold.
\end{enumerate}
Therefore, $((l_A, l_B), \ClockValuation{})$ is the leaf of a Contradiction DAG, which implies that it is the root of a Contradiction DAG with depth 1.

We now show the induction step. In case the state $((l_A, l_B), \ClockValuation{})$ cannot be element of a virtual bisimulation in order n, the induction step holds trivially by the induction hypothesis. Therefore, we assume that the state cannot be element of a virtual bisimulation in order n+1 but it can be element of a virtual bisimulation in order n. Due to the base case, we are allowed to assume $n \geq 1$. Under these circumstances, Definition~\ref{def:app:DAG-are-correct:bounded-virtual-bisimulation} implies that one of the following three statements hold:
\begin{enumerate}
\item Either there exists an $d \in \TimeDomain$ such that there exists a transition $((l_A, l_B), \ClockValuation{}) \STLTSTrans{d} ((l_A, l_B), \ClockValuation{} + d)$ and $((l_A, l_B), \ClockValuation{} + d)$ cannot be element of a virtual bisimulation in order n or
\item there exists a transition $l_A \TATrans{g_A}{\sigma}{R_A} l_{A, 1}$ with $\ClockValuation{} \models g_A$ and $[R_A \rightarrow 0] \ClockValuation{} \models I(l_{A, 1})$ and for any transition $((l_A, l_B), \ClockValuation{}) \STLTSTrans{\sigma} ((l_{A, 1}, l_{B, 1}), \ClockValuation{1})$ with $\forall c \in C_A : [R_A \rightarrow 0]\ClockValuation{}(c) = \ClockValuation{1}(c)$ 
\begin{enumerate}
\item either $((l_{A, 1}, l_{B, 1}), \ClockValuation{1})$ is synchronized but cannot be element of a virtual bisimulation in order $n$ or
\item $((l_{A, 1}, l_{B, 1}), \ClockValuation{1})$ is not synchronized and the target of the outgoing sync-transition $((l_{A, 1}, l_{B, 1}), \ClockValuation{1}) \STLTSTrans{\text{sync}} ((l_{A, 1}, l_{B, 1}), \ClockValuation{1, s})$ cannot be element of a virtual bisimulation in order $n$ or
\end{enumerate}
\item any analog condition for $B$ holds.
\end{enumerate}
We only consider the first two cases, as the third case can be shown analogously. 

If the first case holds, we can create a Contradiction DAG by creating the Contradiction DAG of maximum depth 2n with $((l_A, l_B), \ClockValuation{} + d)$ as root (which is possible by the induction hypothesis) and then adding the transition $((l_A, l_B), \ClockValuation{}) \STLTSTrans{d} ((l_A, l_B), \ClockValuation{} + d)$ to it. Therefore, the induction step holds in this case.

If the second case holds, we create the Contradiction DAG by iterating through all transitions $l_B \TATrans{g_B}{\sigma}{R_B} l_{B, 1}$ and adding the corresponding transition $((l_A, l_B), \allowbreak \ClockValuation{}) \STLTSTrans{\sigma} ((l_{A, 1}, l_{B, 1}), R_A(R_B(\ClockValuation{})))$ to $E$, in case it exists in $\STLTSTrans{}$. We distinguish two cases:
\begin{enumerate}
\item If $((l_{A, 1}, l_{B, 1}), R_A(R_B(\ClockValuation{})))$ is synchronized, we know by the precondition that this state cannot be element of a virtual bisimulation in order $n$. Therefore, we can create a Contradiction DAG with maximum depth 2n with this state as root by the induction hypothesis.
\item If $((l_{A, 1}, l_{B, 1}), R_A(R_B(\ClockValuation{})))$ is not synchronized, we add the $\text{sync}$-transition $((l_{A, 1}, l_{B, 1}), R_A(R_B(\ClockValuation{}))) \STLTSTrans{\text{sync}} ((l_{A, 1}, l_{B, 1}), \ClockValuation{1, s})$ to $E$ and by the precondition and the induction hypothesis, we can build a Contradiction DAG with $((l_{A, 1}, l_{B, 1}), \ClockValuation{1, s})$ as root with maximum depth 2n.
\end{enumerate}
Since the added transitions obviously fulfill the conditions of the second case of Definition~\ref{def:main:new-def-virt-bisim:DAG}, we can build a Contradiction DAG with maximum depth 2(n+1) in each case and, therefore, the induction step holds.

Since the induction hypothesis holds and the initial state cannot be element of a virtual bisimulation if and only if there exists a finite value $m$ such that it cannot be element of a virtual bisimulation in order $m$, the proposition holds.
\end{proof}

Proposition~\ref{prop:app:DAG-are-correct:if-ct-then-not-bisim} and Proposition~\ref{prop:app:DAG-are-correct:if-not-bisim-then-ct} imply Proposition~\ref{prop:main:virtual-bisimulation:contradiction-DAGs-only-for-non-bisimilar}.

\section{Construction of a Contradiction DAG}%
\label{app:create-cont-dag}

We show Proposition~\ref{prop:main:contradiction:find-a-DAG} here. We first show that there finite many elements reachable using the procedure.

\begin{proposition}
\label{prop:app:create-cont-dag:finite}
The number of elements reachable by the procedure described in Proposition~\ref{prop:main:contradiction:find-a-DAG} is finite.
\end{proposition}

\begin{proof}
Since the number of regions is finite~\cite{ALUR1994183}, the number of location-equivalent states with clock valuations that are not region-equivalent is also finite.
\end{proof}

We now show that if the procedure returns a DAG (and does not terminate due to the fact that all possible alternatives are tried and no result has been found), then it is a Contradiction DAG.

\begin{proposition}
\label{prop:app:create-cont-dag:valid-implies-cont-dag}
Eventually, all leaves of the DAG created by the procedure described in Proposition~\ref{prop:main:contradiction:find-a-DAG} can serve as leaves of a Contradiction DAG and the returned result is a Contradiction DAG.
\end{proposition}

\begin{proof}
To show this proposition, we have to show that for all states in the DAG Definition~\ref{def:main:new-def-virt-bisim:DAG} holds. We show this statement by induction in the number of maximum steps to a leaf. Obviously, all synchronized states fulfill a contradiction of the corresponding symbolic state by construction of the procedure.\\
\textbf{Hypothesis: } Any state of the DAG fulfills the conditions of Definition~\ref{def:main:new-def-virt-bisim:DAG}. \\
\textbf{Base Case: } If the state is a leaf, then it fulfills the conditions of Definition~\ref{def:main:new-def-virt-bisim:DAG}. \\
\textbf{Induction Step: } If the hypothesis holds for any state that has no more than $n$ steps to any leaf reachable from this state within the DAG, then the hypothesis also holds for any state that has no more than $n+1$ steps to any leaf reachable from this state within the DAG.

The base case holds trivially by construction. To show the induction step, we consider a state $((l_A, l_B), \ClockValuation{})$ that has no more than $n+1$ steps to any leaf reachable from this state within the DAG. If the maximum number of steps is less than $n+1$, the statement holds trivially by the induction hypothesis. If the maximum number of steps is $n+1$, it cannot be a leaf and we distinguish the following cases:
\begin{enumerate}
\item If $((l_A, l_B), \ClockValuation{})$ is not synchronized, the induction step trivially holds by the second condition of the procedure described in Proposition~\ref{prop:main:contradiction:find-a-DAG}.
\item If $((l_A, l_B), \ClockValuation{})$ is synchronized, we know that it fulfills the contradiction of the corresponding symbolic state and, therefore, we know that it cannot be element of a virtual bisimulation. Since it cannot be a leaf of a Contradiction DAG, we know that the corresponding conditions of Definition~\ref{def:main:new-def-virt-bisim:DAG} are not fulfilled. Therefore, this case can be shown analogously to Proposition~\ref{prop:app:DAG-are-correct:if-not-bisim-then-ct}.
\end{enumerate}
\end{proof}

Proposition~\ref{prop:app:create-cont-dag:valid-implies-cont-dag} implies Proposition~\ref{prop:main:contradiction:find-a-DAG}.
\end{subappendices}
\end{document}